\documentclass[11pt]{article}
\usepackage[paperheight=11.7in,paperwidth=8.7in,margin=1in]{geometry}
\usepackage[normalem]{ulem}
\usepackage{amsthm}
\usepackage{bm}
\usepackage{mathtools}
\usepackage{thmtools}
\declaretheoremstyle[headfont=\normalfont]{normalhead}
\usepackage{color}
\usepackage{graphicx}
\usepackage{morefloats}
\usepackage[hidelinks]{hyperref}
\usepackage{amssymb}
\usepackage{textcomp}
\usepackage{url}
\usepackage{float}
\usepackage{appendix}
\usepackage{authblk}
\usepackage{caption}
\usepackage{booktabs}
\usepackage{algorithm2e}

\usepackage[
backend=bibtex,
giveninits=true,
date=year,
style=authoryear,
sorting=nty,
url=false,
eprint=false,
maxnames=2,
maxbibnames=10
]{biblatex}
\DeclareNameAlias{author}{family-given}
\newcommand\R{\mathbb{R}}

\newcommand\E{\mathbb{E}}

\newcommand{\calH}{\mathcal{H}}

\newcommand{\calL}{\mathcal{L}}
\newcommand{\calM}{\mathcal{M}}

\newcommand{\calO}{\mathcal{O}}
\newcommand{\calP}{\mathcal{P}}

\newcommand{\calT}{\mathcal{T}}

\newcommand{\calV}{\mathcal{V}}
\newcommand{\calW}{\mathcal{W}}

\newcommand{\ra}{\rightarrow}

\newcommand{\op}{\operatorname{op}}

\renewcommand\epsilon{\varepsilon}

\newcommand{\vect}[1]{\bm{#1}}

\DeclareMathOperator*{\argmin}{arg\,min}
\DeclareMathOperator*{\esssup}{ess\,sup}

\newcommand{\M}{\calM}

\newtheoremstyle{mydef}
{\topsep}{\topsep}%
{}{}%
{\itshape}{}
{\newline}
{%
  \rule{\textwidth}{0.0pt}\\*%
  \thmname{#1}~\thmnumber{#2}\thmnote{\-\ #3}.\\*[-1.5ex]%
  \rule{\textwidth}{0.0pt}}%
  
\begin{document}

\newtheorem{conjecture}{Conjecture}
\newtheorem{proposition}{Proposition}
\newtheorem{theorem}{Theorem}[section]
\newtheorem{question}{Question}
\newtheorem{remark}{Remark}
\newtheorem{proposal}{Proposal}
\newtheorem{lemma}{Lemma}[section]
\newtheorem{corollary}{Corollary}[section]
\newtheorem{observation}{Observation}[section]
\newtheorem{assumption}{Assumption}[section]

\author[1]{Eardi Lila}

\author[1,2]{Wenbo Zhang}
\author[3]{Swati Rane Levendovszky}
\author[ ]{\\for the Alzheimer's Disease Neuroimaging Initiative\footnote{\noindent Data used in preparation of this article were obtained from the Alzheimer's Disease Neuroimaging Initiative (ADNI) database
(adni.loni.usc.edu). As such, the investigators within the ADNI contributed to the design and implementation of ADNI and/or
provided data but did not participate in analysis or writing of this report. A complete listing of ADNI investigators can be found at:
\url{http://adni.loni.usc.edu/wp-content/uploads/how_to_apply/ADNI_Acknowledgement_List.pdf}}}

\affil[1]{Department of Biostatistics, University of Washington, USA}
\affil[2]{Department of Statistics, University of California, Irvine, USA}
\affil[3]{Department of Radiology, University of Washington, USA}

\date{}

\title{Interpretable discriminant analysis for functional data supported on random nonlinear domains with an application to Alzheimer's disease}
\maketitle
\begin{abstract}
\hskip -.2in
\noindent
We introduce a novel framework for the classification of functional data supported on nonlinear, and possibly random, manifold domains. The motivating application is the identification of subjects with Alzheimer's disease from their cortical surface geometry and associated cortical thickness map. The proposed model is based upon a reformulation of the classification problem as a regularized multivariate functional linear regression model. This allows us to adopt a direct approach to the estimation of the most discriminant direction while controlling for its complexity with appropriate differential regularization. Our approach does not require prior estimation of the covariance structure of the functional predictors, which is computationally prohibitive in our application setting. We provide a theoretical analysis of the out-of-sample prediction error of the proposed model and explore the finite sample performance in a simulation setting. We apply the proposed method to a pooled dataset from the Alzheimer's Disease Neuroimaging Initiative and the Parkinson's Progression Markers Initiative. Through this application, we identify discriminant directions that capture both cortical geometric and thickness predictive features of Alzheimer's disease that are consistent with the existing neuroscience literature.
\end{abstract}

\section{Introduction}\label{sec:introduction}

Functional discriminant analysis, a statistical framework used to predict categorical outcomes from functional predictors, has been extensively studied within the field of Functional Data Analysis (FDA) \parencite{ramsay2015functional,ferraty2006nonparametric,horvath2012inference,hsing2013theoretical} and has motivated a large body of literature \parencite[see, e.g.,][]{james2001functional,muller2005functional,preda2007regression,
delaigle2012achieving,dai2017optimal,berrendero2018use,
kraus2019classification,park2021sparse}. 
However, most of the existing methods are concerned with the classification of functions supported on one-dimensional linear domains, which can be a limiting assumption in many modern biomedical applications \parencite{zhu2023statistical}. On the other hand, recent work on the analysis of functional data with manifold structure has mostly focused on vector-valued functions with non-Euclidean constraints in the image space \parencite[see e.g.,][]{su2014statistical,dai2018principal,lin2017extrinsic,
dubey2020functional,kim2021smoothing}.

In this paper, motivated by the analysis of modern multi-modal imaging data, we propose a novel functional discriminant analysis model that can handle functional predictors supported on nonlinear sample-specific manifold domains, which we term Functions on Surfaces (FoSs) \parencite{lila2020statistical}. An example of such `object data' \parencite{marron2021object} is provided in Figure~\ref{fig:setting}A, illustrating our motivating application of identifying subjects with Alzheimer's disease from FoSs that are subject-specific cortical surfaces coupled with cortical thickness measurements. The statistical analysis of these object data poses unique statistical challenges. This is mainly due to the non-Euclidean structure of each individual domain, which makes it difficult to define appropriate spatial regularization, and due to the more abstract non-Euclidean structure of the latent space where the random domains are supported, which further invalidates traditional linear statistical models \parencite{kendall1984shape,grenander1998computational,
dryden2016statistical,younes2019shapes}. 


\begin{figure}[!htb]
\centering
\includegraphics[width=0.7\textwidth]{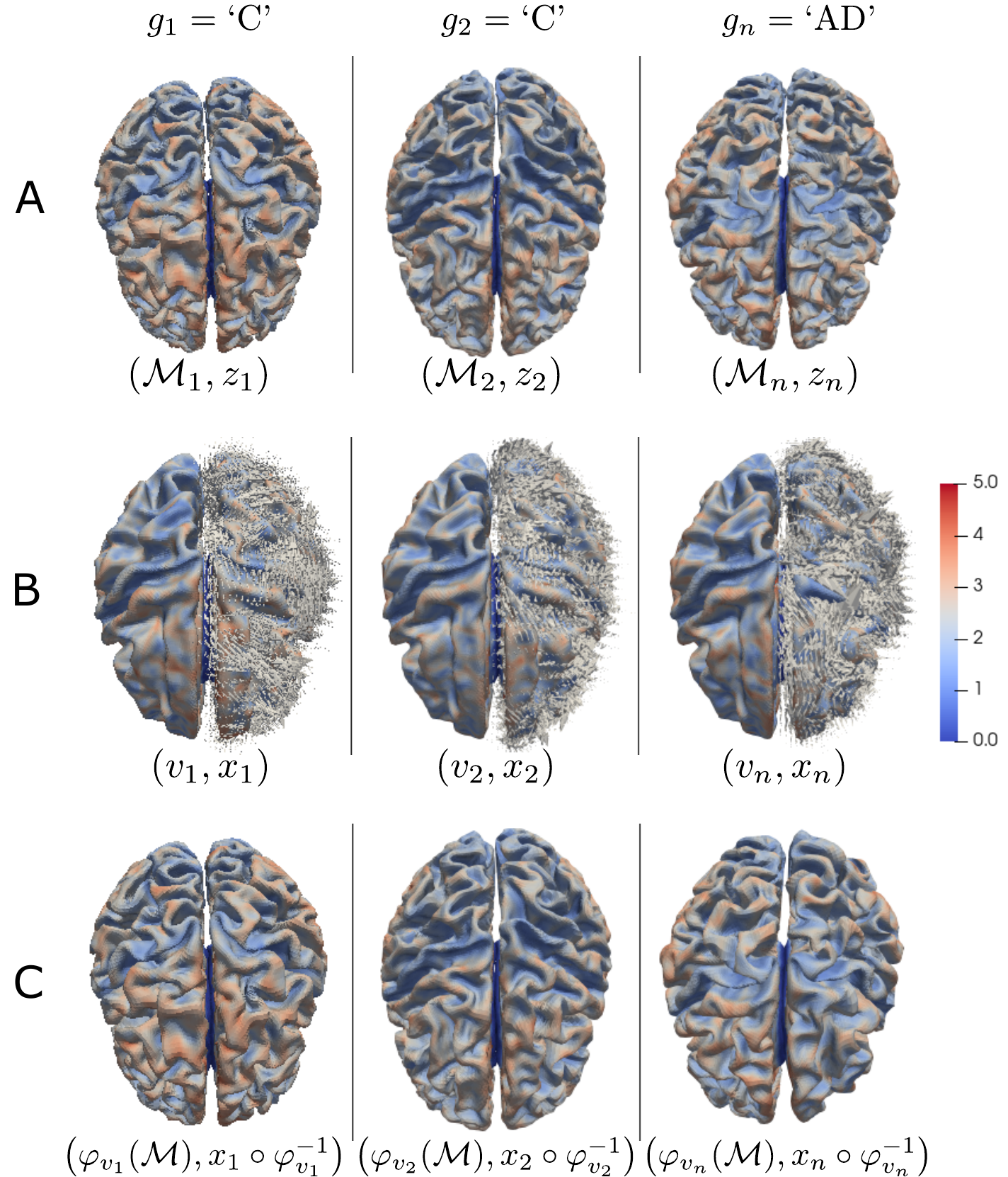}
\caption[]{Panel A: FoSs of three subjects in the training sample, where $g_i \in \{\text{`C'},\text{`AD'} \}$ denotes the disease state of the $i$th individual (C: Control, AD: Alzheimer's Disease), $\M_i$ is a two-dimensional manifold encoding the geometry of the cerebral cortex, and $z_i:\M_i \ra \R$ is a real function, supported on $\M_i$, describing cortical thickness (in mm). Panel B: Linear representation $\left(v_i, x_i \right)$ of each FoS $\left(\M_i, z_i \right)$ shown in Panel A. Here $v_i: \R^3 \ra \R^3$ is a vector-valued function encoding the geometry of the $i$th individual. This is depicted as a collection of 3D vectors $\{v_i(p_j)\}$ for a dense set of points $\{p_j	\} \subset \R^3$. For clarity, the function $v_i$ is displayed only on half of its domain $\R^3$. The function $x_i:\M \ra \R$ describes the spatially normalized thickness map of the $i$th individual on the fixed template $\M$. Panel C: FoS $(\varphi_{v_i} (\M) , x_i \circ \varphi^{-1}_{v_i})$ parametrized by the associated functions $\left(v_i, x_i \right)$ in Panel B. This is a close approximation of the FoS $\left(\M_i, z_i \right)$ in Panel A.

}
\label{fig:setting}
\end{figure}

Current approaches in the literature do not comprehensively address these challenges. Although various methods have been proposed to model functional data on multi-dimensional domains, they often focus on flat domains \parencite{goldsmith2014smooth, wang2017generalized, kang2018scalaronimage, feng2020bayesian, yu2021multivariate} or assume nonlinear but fixed domains \parencite{chung2015unified, chung2021statistical, lila2016smooth, mejia2020bayesian}. 

There has also been considerable work on the simpler setting of random surfaces that are \textit{not} coupled with functional data. These efforts can be broadly grouped into three main approaches. The first approach leverages global parametrizations to represent surfaces, employing either an $L^2$ metric \parencite{chung2008encoding, epifanio2014hippocampal,ferrando2020detecting} or a non-Euclidean metric \parencite{jermyn2012elastic, jermyn2017elastic, kurtek2015comprehensive, zhang2022lesa}. The second approach, which is more closely related to the one adopted in this work, uses diffeomorphic deformation functions of the surfaces' embedding space \parencite{vaillant2004statistics, younes2019shapes, arguillere2016diffeomorphic}, allowing for the inclusion of topological constraints. The third approach, prevalently used in neuroimaging studies, employs pre-specified or spectrum-based descriptors of shape \parencite{reuter2006laplacebeltrami, im2008brain, wachinger2015brainprint, hazlett2017early, wang2017holistic, dong2019applying, dong2020applying}. A critical drawback of the latter approach is the inability to uniquely map the discrete representations back to the original space of random surfaces.

The statistical analysis of random surfaces that are \textit{coupled} with functional data has not been extensively explored. One exception is the model in \textcite{zaetz2015novel} which focuses on annotated surfaces. In addition, unsupervised models have been investigated by \textcite{charlier2017fshape, lila2020statistical}. \textcite{lee2017atlasbased} have dealt with the classification problem by employing the fshape framework \parencite{charlier2017fshape} to represent the data and by using a linear discriminant model on the resulting finite-dimensional representations. In contrast, the statistical framework for the discriminant analysis of FoSs presented in this paper avoids any dimension reduction of the predictors and instead employs spatial penalties to regularize the discriminant direction. It tackles the non-Euclidean nature of the latent space of random domains by defining appropriate linear functional representations of FoSs, effectively reframing the problem of classifying FoSs as the problem of classifying bivariate functional data supported on general, but fixed, domains. A direct model for the estimation of a functional discriminant direction is then defined on the representation space where differential spatial regularizations are introduced to produce interpretable and well-defined estimates. A key feature of the proposed representation is its invertibility, which enables us to map estimates from the representation space back to the original space. This allows us to explore and interpret the estimated classification rule in the context of the original neurobiological objects of our motivating application.

The rest of the paper is organized as follows. In Section~\ref{sec:lin_repr}, we describe the representation model adopted to parametrize the non-Euclidean space of FoSs using linear function spaces. In Sections~\ref{sec:univariate} and \ref{sec:multivariate}, we develop a novel discriminant analysis model on the parametrizing linear function spaces and provide theoretical guarantees for the prediction performance of the proposed model. We introduce an efficient discretization approach in Section~\ref{sec:discretization} and apply the proposed model to the Alzheimer's Disease Neuroimaging Initiative (ADNI) and Parkinson's Progression Markers Initiative (PPMI) datasets in Section~\ref{sec:application}. Proofs and simulations are left to the appendices.

\section{Functional  data supported on general random domains}\label{sec:lin_repr}

The data considered in this work is a sample of triplets
\begin{equation}\label{eq:data}
\left\{(g_i, \M_i , z_i ), i = 1,\ldots,n\right\},
\end{equation}
where $g_i$ is a binary label, $\M_i \subset \R^3$ is a sample-specific closed two-dimensional manifold embedded in $\R^3$, and $z_i : \M_i \ra \R$ is a scalar function supported on the geometric object $\M_i$. We moreover assume that the points on the geometries $\{\M_i\}$ of the observed FoSs are in one-to-one correspondence across subjects. We refer to the pairs $\{(\M_i,z_i)\}$ as FoSs. 

In Figure~\ref{fig:setting}A, we display three observations of the training sample of our final application, where $g_i$ encodes the disease state of the subject, $\M_i$ encodes the geometry of the cerebral cortex, and $z_i$ encodes the cortical thickness map supported on $\M_i$. Our goal is to build a classifier from the given training sample that can predict the binary label $g^*$ of a previously unseen FoS $(\M^*,z^*)$.

\subsection{Linear functional representations}
In our motivating application, the geometric objects $\{\M_i: i=1,\ldots,n\}$ are topologically equivalent to a sphere, and therefore, they do not display holes or self-intersections. To inform our model of such physical non-Euclidean constraints, we define a convenient unconstrained representation model for the FoSs $\{(\M_i,z_i): i=1,\ldots,n\}$ in terms of objects belonging to linear function spaces.

Let $\M$ be a template two-dimensional manifold embedded in $\R^3$ that is topologically equivalent to a sphere.
We denote by $\calL^2(\M)$ the space of square integrable functions over $\M$, equipped with the standard inner product $\langle \cdot, \cdot \rangle_{\calL^2(\M)}$ and norm $\|\cdot\|_{\calL^2(\M)}$, and denote by $\calL^2(\R^3, \R^3)$ the space of square integrable vector-valued functions from $\R^3$ to $\R^3$, with inner product $\langle \cdot, \cdot \rangle_{\calL^2(\R^3, \R^3)}$ and norm $\| \cdot \|_{\calL^2(\R^3, \R^3)}$. Let $\calV(\R^3) \subset \calL^2(\R^3, \R^3)$ be a Reproducing Kernel Hilbert Space (RKHS) of smooth functions with compact support in $\R^3$. We then introduce a diffeomorphic operator $\varphi$ such that $\varphi_{v}: \R^3 \ra \R^3$ is a diffeomorphic function for every choice of $v \in \calV(\R^3)$ \parencite{younes2019shapes}. We denote by $\varphi_{v}(\M)$ the geometric object given by displacing every point $p \in \M \subset \R^3$ to the new location $\varphi_{v}(p)$. A direct consequence of $\varphi_{v}$ being diffeomorphic is that $\varphi_{v}(\M)$ is topologically equivalent to a sphere for every choice of $v \in \calV(\R^3)$. The construction of the diffeomorphic operator adopted in this paper and the computation of $v_i \in \calV(\R^3)$ such that $\varphi_{v_i}(\M) \approx \M_i$ are detailed in Appendix~\ref{sec:diff_op}. 

Next, we use the estimated $v_i$, and $z_i$, to define the spatially normalized function $x_i: \M \to \R$ as $x_i = z_i \circ \varphi_{v_i}$. For each point $p \in \M$, $x_i(p)$ is given by $z_i \circ \varphi_{v_i}(p)$, that is, the value of $z_i$ at the corresponding point $\varphi_{v_i}(p) \in \M_i$. We can then represent each FoS $(\M_i , z_i)$ with a pair of functions $(v_i, x_i)$ such that
\begin{equation}\label{eq:FoS_repr}
(\M_i , z_i) \approx (\varphi_{v_i} (\M) , x_i \circ \varphi^{-1}_{v_i}),
\end{equation}
where $v_i \in \calV(\R^3)$ and $x_i \in \calL^2(\M)$. 

We can depict the representation model introduced as follows:
\begin{equation}
(v_i, x_i) \longleftrightarrow (\M_i , z_i),
\end{equation}
meaning that given a FoS $(\M_i , z_i)$, we can compute a loss-less representation $(v_i, x_i)$ as described earlier, and vice-versa, given the representation $(v_i, x_i)$, we can compute the associated FoS through equation (\ref{eq:FoS_repr}). Hence, the pair of functions $(v_i, x_i)$ provides us with a linear representation of the original FoS $(\M_i , z_i)$ where every geometric object $\M_i$ is modeled as a (diffeomorphic) deformation of the template, i.e. $\varphi_{v_i} (\M)$, while the associated function $z_i$ is given by `transporting' the spatially normalized function $x_i$ onto $\M_i$ with such a deformation.

The approach described allows us to recast the original non-Euclidean problem of learning a classifier from the training sample $\{(g_i, \M_i, z_i) | i = 1,\ldots,n \}$ as the problem of learning a classifier from 
\begin{equation}\label{eq:repr}
\left\{(g_i, v_i, x_i)| i = 1,\ldots,n \right\},
\end{equation}
where $v_i \in \calV(\R^3) \subset \calL^2(\R^3, \R^3)$ and $x_i \in \calL^2(\M)$ are two functional predictors both belonging to \textit{linear} function spaces. In Figure~\ref{fig:setting}B, we show the functional linear representations associated with the three FoSs in Figure~\ref{fig:setting}A.

Crucially, the representation mapping employed here is `invertible', meaning that any pair of estimates $(\beta^G, \beta^F) \in \calV(\R^3) \times \calL^2(\M)$, such as the `direction' that optimally discriminates between two classes, can be mapped back to the original space of FoS using equation (\ref{eq:FoS_repr}). This mapping defines the associated trajectories of FoSs 
\begin{equation}\label{eq:FoS_traj}
\left\{ \left( \varphi_{c_1\beta^G}(\M), c_2 \beta^F \circ \varphi_{c_1\beta^G}^{-1} \right), c_1,c_2 \in \R \right\},
\end{equation}
where $\varphi_{c_1\beta^G}(\M)$ is guaranteed to be topologically equivalent to a sphere, thereby satisfying the physical constraints of the problem considered. 

In contrast to methods that require computing shape features, such as the spectrum of the Laplace-Beltrami operator \parencite{reuter2006laplacebeltrami,wachinger2015brainprint}, the approach adopted here provides us with interpretable discriminant directions in the space of the original neurobiological objects, as described by equation (\ref{eq:FoS_traj}). In addition, unlike approaches that work with global parametrizations of FoSs \parencite[see, e.g.,][]{chung2008encoding, epifanio2014hippocampal, zaetz2015novel, ferrando2020detecting}, the representation model used in this study is independent of the imaging data type, as long as we can specify how $\varphi_{v_i}$ deforms our objects and a suitable similarity measure. Therefore, the framework proposed in this work has the potential to accommodate additional types of data, such as streamlines generated from diffusion tensor images, where there may not be a one-to-one correspondence across subjects, but for which optimal transport similarity measures have been developed \parencite{feydy2017optimal}. Although our work assumes that the FoSs are in one-to-one correspondence, this is not strictly necessary. Assuming a one-to-one correspondence simplifies the definition of a similarity measure and makes it easier to compute these representations for complex objects such as cortical surfaces, as detailed in Appendix~\ref{sec:diff_op}. In Section~\ref{sec:application}, we compare the performance of our representation model with alternative models, in the context of our motivating application.

\subsection{Discriminant analysis on the parametrizing linear function spaces}
The aim of Sections~\ref{sec:univariate} and \ref{sec:multivariate} is to provide methodology for learning a linear classifier, from the training data $\{(g_i, v_i, x_i)| i = 1,\ldots,n\}$ displayed in Figure~\ref{fig:setting}B, by introducing a novel functional classification model that has the following crucial characteristics:
\begin{itemize}
\item Does not rely on Functional Principal Components Analysis (FPCA), or related dimension reduction methods, to reduce the dimension of the functional predictors, bypassing the intrinsic assumption that the discriminant direction is well represented by the space spanned by the first few unsupervised PC functions;

\item Can be applied to bivariate, and possibly multivariate, functional predictors each supported on a different domain;

\item Allows for explicit spatial regularization of the estimates on potentially nonlinear manifold domains, yielding well-defined and interpretable estimates;

\item Provides a direct approach to estimating the discriminant directions without relying on prior computation of the covariance structure, which is prohibitive in our application setting.
\end{itemize}

\section{Linear discriminant analysis over general domains}\label{sec:univariate}

We begin by focusing on the sub-problem of defining a classifier for the training data $\{(g_i,x_i)\}$, i.e., for the spatially normalized functional predictors supported on the fixed nonlinear manifold $\M$. In Section~\ref{sec:multivariate}, we then extend the proposed model to account for the geometric component $v_i$, in an additive fashion.

Assume the training set $\{(g_i,x_i)\}$ consists of $n$ independent copies of $(G,X)$, a pair of random variables with $X$ a zero-mean random function taking values in $\calL^2(\M)$ and $G$ a binary random variable such that $P\left(G = 1\right) = \pi_1$ and $P\left(G = 2\right) = \pi_2$. Let $\mu_1 = \E \left[X|G = 1 \right]$ and $\mu_2 = \E \left[X|G = 2 \right]$ denote the conditional means of $X$ and assume $\mu_1 \neq \mu_2$. Moreover, let $C(p,q) = \E\left[X(p)X(q)\right] , p, q \in \M$ denote the covariance function of $X$ and assume this is square integrable, i.e.,  $\int_\M \int_\M C(p, q)^2 \,dp \,dq < \infty$. For a real, symmetric, square-integrable, and non-negative function $R \in \calL^2(\M \times \M)$, let the integral operator $L_R:\calL^2(\M) \ra \calL^2(\M)$ be defined as
\[
L_{R}(\beta)(\cdot) = \int_\M R(p, \cdot )\beta(p) \,dp, \, \forall \beta \in \calL^2(\M).
\] 
Consequently, $L_C$ denotes the covariance operator of $X$, which is a compact self-adjoint operator and therefore admits the following spectral representation 
\begin{equation}
L_C(\beta) = \sum_{k=1}^\infty \theta_k \langle \beta, e_k \rangle_{\calL^2(\M)} e_k,
\end{equation}
in terms of the eigenvalues $\theta_1 \geq \theta_2 \geq \ldots \geq 0$ and associated eigenfunctions $e_1, e_2, \ldots \subset \calL^2(\M)$ of $L_C$.

Let $L_C^{-1}$ denote the linear inverse covariance operator, where $L_C^{-1}(e_k) = \frac{1}{\theta_k} e_k$ for all $k \geq 1$. Assume that $\|L_C^{-1}(\mu_2 - \mu_1)\|_{\calL^2(\M)} < \infty$ and define the population quantity $\beta^0 \in {\calL^2(\M)}$ such that
\[
L_C \beta^0 = \mu_2 - \mu_1.
\]
Note that this is an assumption on the underlying population quantities and will not have practical implications. However, it allows us to have a unique well-defined variable $\beta^0$ to study the convergence properties of the proposed model. For a discussion on the case $\|L_C^{-1}(\mu_2 - \mu_1)\|_{\calL^2(\M)} = \infty$, which is related to the perfect classification phenomenon, see \textcite{delaigle2012achieving, berrendero2018use, chen2018sensible, kraus2019classification}. 

The function $\beta^0$ can be understood as a functional analog of the multivariate discriminant vector of a linear discriminant analysis \parencite{shin2008extension}. For a new observation with predictor $x^* \in \calL^2(\M)$, it can be used to predict the associated label $g^*$ with the linear classification rule $\langle \beta^0, x^* \rangle_{\calL^2(\M)} > c^\text{th}$, with $c^{\text{th}}$ an appropriately chosen threshold. Moreover, if $X$ is a Gaussian random function within each group in $G$, it can be shown that the function $\beta^0$ defines the linear classifier that minimizes the misclassification error rate, and it is therefore optimal \parencite{delaigle2012achieving}. The discriminant direction $\beta^0$ can also be equivalently defined as the minimizer of the functional
\begin{equation}\label{eq:min_pop}
\frac{1}{2}\langle \beta, L_C \beta \rangle_{\calL^2(\M)} - \langle \mu_2 - \mu_1, \beta \rangle_{\calL^2(\M)}.
\end{equation}

In practice, the population quantities $C$, $\mu_1$, and $\mu_2$ are unknown and need to be estimated from the data. The goal of a classification model is therefore to recover $\beta^0$ from the training sample $\{(g_i,x_i): i=1,\ldots,n\}$ of $n$ independent copies of $(G,X)$. This can be achieved by using the sample covariance $L_{\hat{C}}$, with 
\[
\hat{C}(p,q) = \frac{1}{n}\sum_{i=1}^{n} x_i(p)x_i(q),\qquad p,q \in \M
\]
and the sample conditional means $\hat \mu_1$ and $\hat \mu_2$ to replace the population counterparts in equation (\ref{eq:min_pop}). An estimate $\hat \beta$ of $\beta^0$ can then be defined as a minimizer of 
\begin{equation}\label{eq:pen_cov_est}
\frac{1}{2}\langle \beta, L_{\hat{C}} \beta \rangle_{\calL^2(\M)} - \langle \hat \mu_2 - \hat \mu_1, \beta \rangle_{\calL^2(\M)} + \calP(\beta),
\end{equation}
where a penalty term $\calP(\beta)$ is typically added to overcome the ill-posedness of the minimization problem, which is due to the low-rank structure of $L_{\hat{C}}$. For instance, \textcite{park2021sparse} define a penalty $\calP$ that encourages the estimate $\hat \beta$ to be smooth and sparse, while \textcite{kraus2019classification} consider a ridge-type penalty. 

The functional discriminant model in equation (\ref{eq:pen_cov_est}) requires precomputing the empirical covariance function, which is generally not possible for dense functional data supported on multidimensional domains and is ultimately not feasible in our application setting. We therefore propose a direct regularized approach to estimating $\beta^0$. This will be possible thanks to the following simple observation. As noted for instance in \textcite{delaigle2012achieving}, the discriminant direction $\beta^0$ can be equivalently characterized as the solution to the minimization problem
\begin{equation}\label{eq:exp_ls}
\beta^0 = \argmin_{\beta}  \E\left[ Y - \langle X,\beta \rangle_{\calL^2(\M)} \right]^2,
\end{equation}
where $Y$ is an auxiliary scalar random variable such that $Y = -\frac{1}{\pi_1}$ if $G = 1$ and $Y = \frac{1}{\pi_2}$ otherwise. In other words, the classification problem considered can be reformulated as a functional regression problem. This motivates the adoption of a least-squares approach to estimating $\beta^0$, based on the empirical counterpart of the objective function in equation (\ref{eq:exp_ls}), where an additional differential regularization term is introduced to incorporate information on the geometric domain $\M$ and overcome the ill-posedness of the problem. For multivariate data, analogous least-squares formulations have also been adopted, for instance, in \textcite{hastie1994flexible,mai2012direct,gaynanova2020prediction}.

\subsection{Regularized estimation}
Given the training sample $(g_i, x_i)$, introduce a scalar variable $y_i$ such that $y_i = - \frac{n}{n_1}$ if $g_i = 1$ and $y_i = \frac{n}{n_2}$ otherwise, where $n_1$ and $n_2$ represent the sample sizes of class $1$ and $2$, respectively. Observe that $-\frac{n}{n_1}$ and $\frac{n}{n_2}$ are estimates of the values that can be taken by the random variable $Y$. Let $\calW^2(\M)$ be the Sobolev space of functions in $\calL^2(\M)$ with first and second distributional derivatives in $\calL^2(\M)$. We define an estimate $\hat \beta \in \calW^2(\M)$ of the population quantity $\beta^0$ as the solution to the following minimization problem
\begin{align}\label{eq:model_ls_univ}
\hat \beta = \argmin_{\beta \in \calW^2(\M)}  \frac{1}{n} \sum_{i=1}^n \left( y_i -  \langle x_i, \beta \rangle_{\calL^2(\M)} \right)^2 +  \lambda J(\beta), 
\end{align}
where the first term is a least-squares estimate of the objective function in equation (\ref{eq:exp_ls}) and the second term is a differential regularization term.  The parameter $\lambda$ controls the trade-off between the least-squares term of the objective function and the penalty term. Our choice of the regularization term is
\begin{equation}
J(\beta) = \|\Delta_{\M} \beta\|^2_{\calL^2(\M)} + \epsilon \|\beta\|^2_{\calL^2(\M)},
\end{equation}
with $\epsilon \geq 0$. This is a linear combination of two terms. The first one is based on the Laplace-Beltrami operator $\Delta_{\M}: \calW^2 \subset \calL^2(\M)\ra \calL^2(\M)$ and quantifies the smoothness of the function $\beta: \M \ra \R$ on the nonlinear manifold domain $\M$. Specifically, it allows the model estimate $\hat \beta$, at any fixed point $p \in \M$, to borrow strength from the other points on $\M$ while constraining the `information' to propagate coherently with the nonlinear manifold structure of the anatomical object $\M$. The second term is a generic shrinkage-type regularization.

It is worth noting that the function space $\calL^2(\M)$ is linear, even though each function $f \in \calL^2(\M)$ is supported on a nonlinear domain. For Euclidean domains, it is common to define a smooth subspace of $\calL^2(\M)$ by forming an RKHS from a positive-definite kernel. However, constructing a positive-definite kernel that is compatible with the geometry of a manifold is a non-trivial task \parencite{feragen2015geodesic, jayasumana2015kernel}. To overcome this challenge, we constructively define a Sobolev norm $J^{1/2}(\cdot)$ and an associated Sobolev space $\calW^2(\M)$ by leveraging a local differential operator, namely the Laplace-Beltrami operator. The discrete counterpart of this local operator is a sparse matrix, reducing our problem to sparse linear algebra and enabling us to solve equation (\ref{eq:model_ls_univ}) for the large data of our final application. We provide more details about the relationship of our approach with the RKHS approach in Section~\ref{sec:ref_pen}.

\subsection{Theory}
The aim of this section is to provide theoretical guarantees for the performance of the proposed model. Specifically, we provide a probability bound for the out-of-sample risk, i.e., the random variable
\begin{equation} \label{eq:oos_pred}
\E^* \left[ \langle X^*, \beta^0 - \hat \beta \rangle_{\calL^2(\M)} \right] ^2,
\end{equation}
where $X^*$ is a copy of $X$ that is independent of the training data and $\E^*$ is the expectation taken over $X^*$. Equation (\ref{eq:oos_pred}) measures the discrepancy between the prediction made with the estimated parameter $\hat \beta$ and the `optimal' prediction made with the unknown population quantity $\beta^0$.

Assume for simplicity that $\epsilon > 0$. Then, thanks to the Sobolev embedding theorem \parencite{brezis2011functional}, $\exists M \geq 0$ such that for any $p \in \M$
\[
f(p) \leq \sup_q |f(q)| \leq  M \left( \|\Delta_{\M} f\|^2_{\calL^2(\M)} + \epsilon \|f\|^2_{\calL^2(\M)} \right)^{1/2}, \qquad \forall f \in \calW^2(\M),
\]
that is, the evaluation operator is a continuous functional. A direct consequence is that the space $ \calW^2(\M)$ equipped with the norm $J^{1/2}(\cdot) = \left( \|\Delta_{\M} \cdot\|^2_{\calL^2(\M)} + \epsilon \|\cdot\|^2_{\calL^2(\M)} \right)^{1/2}$ is an RKHS with a symmetric, positive definite kernel function $K_\M: \M \times \M \ra \R$ \parencite{berlinet2004reproducing}. The kernel function $K_\M$ is used only for theoretical investigation here and obtaining its explicit form is, in general, not computationally feasible and not necessary. We will, however, take advantage of the fact that $L_{K_{\M}}^{1/2}(\calL^2(\M)) = \calW^2(\M)$, where $L_{K_{\M}}^{1/2}$ denotes the square root of $L_{K_{\M}}$ \parencite{cucker2002mathematical}. For $\epsilon = 0$, the functional $J^{1/2}$ defines a semi-norm rather than a norm and similar arguments hold by restricting ourselves to the subspace of $\calL^2(\M)$ that is orthogonal to the null space of $J^{1/2}$.

Next, we define the sandwich operator $T = L_{K_{\M}}^{1/2} L_{C} L_{K_{\M}}^{1/2}$ \parencite{cai2012minimax} and make the following assumptions.

\begin{assumption}\label{assm:bounded}
The constant $\kappa^2$, defined as $\kappa^2 = \esssup \|L^{1/2}_{K_\M} X\|^2_{\calL^2(\M)}$ is finite.
\end{assumption}

\begin{assumption}\label{assm:smooth}
There exists a smooth function $\beta^0 \in \calW^2(\M)$ such that $\beta^0 = L_C^{-1}(\mu_2 - \mu_1)$.
\end{assumption}

\begin{assumption} \label{assm:effective}
The effective dimension of $T$ satisfies $D(\lambda) = \operatorname{Tr}((T + \lambda I)^{-1}T) \leq c \lambda^{-\theta}$ for constants $c,\theta > 0$. Here $\operatorname{Tr}$ denotes the trace operator.
\end{assumption}

Assumption~\ref{assm:bounded} allows us to use a Hoeffding-type inequality for Hilbert space valued random elements and has no practical implications. This condition is met, for example, when $\|X\|_{\calL^2(\M)}$ is bounded. However, it is more general given that $L^{1/2}_{K_\M} X$ represents a smoothed version of $X$. Assumption~\ref{assm:smooth} guarantees that the population quantity $\beta^0$ is well-defined and belongs to the space of smooth functions $\calW^2(\M)$. Assumption \ref{assm:effective} is expressed in terms of properties of the effective dimension $D(\cdot)$. For our final choice of $\lambda$, it is straightforward to check that this assumption holds by assuming that the eigenvalues $\{ \tau_k \}$ of $T$ decay as $\tau_k \asymp k^{-2r}$, with $r > \frac{1}{2}$. This is a typical assumption in the literature on functional linear models \parencite{cai2012minimax} and is related to the rate of decay of the eigenvalues of $L_{K_\M}$ and $L_C$, and their alignment.

The following theorem provides an upper bound for the out-of-sample risk.
\begin{theorem}\label{thm:risk_manifold}
Under Assumptions~\ref{assm:bounded}-\ref{assm:effective}, if $\lambda \asymp n^{-\frac{1}{1+\theta}}$, the estimator $\hat \beta$ in equation (\ref{eq:model_ls_univ}) is such that
\begin{align}
\E^* \left[ \langle X^*, \beta^0 - \hat \beta \rangle_{\calL^2(\M)} \right] ^2 = \calO_p \left( n^{-\frac{1}{1+\theta}} \right).
\end{align}
\end{theorem}%
Similar rates of convergence have been shown to hold for regularized estimates in the functional linear regression setting \parencite[see, e.g., ][]{cai2012minimax,tong2018analysis,sun2018optimal,reimherr2018optimal}. However, a key difference in our model is that the residual random variable $\epsilon = Y - \langle X, \beta^0 \rangle_{\calL^2(\M)}$ and the functional predictor $X$ are not independent, which prevents the direct application of such results. Therefore, Theorem \ref{thm:risk_manifold} shows that in spite of such a dependence structure we are nevertheless able to recover the functional linear model rates of convergence. The proof is provided in Appendix~\ref{sec:proofs}.

\subsection{Nonlinear extensions}\label{sec:non_linear}
To incorporate nonlinearity into the model described in equation (\ref{eq:model_ls_univ}), one can substitute the term $\langle x_i, \beta \rangle_{\calL^2(\M)}$ with a nonlinear function of the data, such as a polynomial term or a single-index model, as done in the context of functional regression models in \textcite{yao2010functional,jiang2011functional}, respectively. However, these extensions come at the cost of estimating additional functional parameters or optimizing a more complex objective function.

If the covariance structures of the two classes are believed to be different, the proposed functional linear discriminant model can be generalized to an approximate functional quadratic discriminant model, following the approach proposed by \textcite{gaynanova2019sparse}, as follows. We estimate the discriminant rule by minimizing the following objective function with respect to $\beta_1,\beta_2 \in \calL^2(\M)$:
\[
\frac{1}{n_1} \sum_{i|g_i = 1} \left( 1 -  \langle x_i, \beta_1 \rangle_{\calL^2(\M)} \right)^2 + \frac{1}{n_2} \sum_{i|g_i = 2} \left( 1 +  \langle x_i, \beta_2 \rangle_{\calL^2(\M)} \right)^2 + \lambda_1 J(\beta_1) + \lambda_2 J(\beta_2),
\]
where $\lambda_1, \lambda_2 > 0$ are tuning parameters. A modified version of Fisher's criterion \parencite{gaynanova2019sparse} is then employed to assign the class by first projecting the data along the estimated directions.

As expected, the simulations presented in Appendix~\ref{sec:simulations} demonstrate that the approximate functional quadratic discriminant model outperforms the functional linear discriminant model when the covariance structures of the two classes differ. Examining the theoretical properties of this extension is beyond the scope of this paper and is left to future work.

\section{Additive multivariate generalizations}\label{sec:multivariate}
We now consider a bivariate extension of the functional model introduced in Section~\ref{sec:univariate}, which incorporates the geometric component of the original data. We therefore consider the training sample $\{(g_i, v_i, x_i)\}$, where $v_i \in \calV(\R^3)$ is a vector field representing the subject-specific geometry of the $i$th subject. Recall that $\left( \calV(\R^3), \| \cdot \|_{\calV(\R^3)} \right)$ is an RKHS of smooth functions with compact support in $\R^3$. Moreover, denote by $K_{\R^3}$ its reproducing kernel.

Suppose the training set $\{(g_i, v_i, x_i)\}$ consists of $n$ independent copies of $(G,V,X)$, a triplet of random variables with $V$ a zero-mean random function taking values in $\calV(\R^3)$, $X$ a zero-mean random function taking values in $\calL^2(\M)$, and $G$ a binary random variable such that $P\left(G = 1\right) = \pi_1$ and $P\left(G = 2\right) = \pi_2$. 

We now adopt the multivariate functional data notation from \textcite{happ2018multivariate}, and define $\vect{X}(p) = \left(V(p_1), X(p_2)\right)$, with $p = (p_1,p_2) \in D = D_1 \times D_2 = \R^3 \times \M$. The multivariate random function $\vect{X}$ takes values in a Hilbert space $\calH = \calL^2(\R^3,\R^3) \times \calL^2(\M)$ with inner product 
$\langle \vect f, \vect g \rangle_{\calH} = \langle f^{(1)}, g^{(1)} \rangle_{\calL^2(\R^3,\R^3)} + \langle f^{(2)}, g^{(2)} \rangle_{\calL^2(\M)}$ for $f, g \in \calH$. Here $f^{(j)}$, with $j \in \{1,2\}$, denotes the $j$th functional component of the multivariate function $\vect f$. For $p, q \in D$, define the matrix of covariances $\vect C(p, q) = \E(\vect{X}(p) \otimes \vect{X}(q))$ with elements $C_{lj}(p_l, q_j) = \E [ X^{(l)}(p_l) X^{(j)}(q_j)]$ where $p_l \in D_l, q_j \in D_j$, $l \in \{1,2\}$, and $j \in \{1,2\}$. Denote the conditional means of $\vect X$ by $\vect \mu_1 = \left(\mu_1^{(1)},\mu_1^{(2)}\right) : = \left(\E \left[V|G = 1 \right], \E \left[X|G = 1 \right] \right)$ and $\vect \mu_2 = \left(\mu_2^{(1)},\mu_2^{(2)} \right) : = \left(\E \left[V|G = 2 \right],\E \left[X|G = 2 \right] \right)$, and assume $\vect \mu_1 \neq \vect \mu_2$. The covariance operator $L_{\vect C}: \calH \ra \calH$ is such that the $j$th component of $L_{\vect C} \vect f$, for any $\vect f \in \calH$, is given by
\begin{equation}
(L_{\vect C} \vect f)^{(j)}(p_j) = \sum_{i=1}^2 \int_{D_i} C_{ij}(q_i,p_j) f^{(i)}(q_i) \, dq_i.
\end{equation}
Similar to the univariate case, we assume that the population quantity $\vect \beta^0 \in \calH$ is well-defined and satisfies the equation
\[
L_{\vect C} \vect \beta^0 = \vect \mu_2 - \vect \mu_1.
\]
This can be viewed as a multivariate generalization of the linear discriminant direction defined in the previous section. We now turn to the problem of defining an estimator for $\vect \beta^0$.

\subsection{Regularized estimation}\label{sec:multivariate-est}
Let the variable $y_i$ be such that $y_i = - \frac{n}{n_1}$ if $g_i = 1$ and $y_i = \frac{n}{n_2}$ otherwise. We define the multivariate functional estimate $\vect{\hat \beta} = \left( \hat \beta^{G}, \hat \beta^{F}\right)$ of the population quantity $\vect{\beta^0}$ to be the solution to the following minimization problem
\begin{equation}\label{eq:model_ls_biv}
\left( \hat \beta^{G}, \hat \beta^{F}\right) =     \argmin_{\substack{\beta^G \in \calV(\R^3)\\ \beta^F \in \calW^2(\M)}}  \frac{1}{n} \sum_{i=1}^n \left( y_i - \langle v_i, \beta^G \rangle_{\calL^2(\R^3,\R^3)} -  \langle x_i, \beta^F \rangle_{\calL^2(\M)} \right)^2 + \lambda_1 \|\beta^G\|_{\calV(\R^3)}^2 + \lambda_2 J(\beta^F),
\end{equation}
with $\lambda_1,\lambda_2$ tuning parameters. 

Equation (\ref{eq:model_ls_biv}) extends the model proposed in Section~\ref{sec:univariate} to account for both the geometric functional descriptor $v_i$ and the function $x_i$ in an additive fashion.  The regularization terms in the equation enforce smoothness on the functional estimates $\hat \beta^{G}$ and $\hat \beta^{F}$ in their respective function spaces.

\subsection{Differential regularization and kernel penalty: a unified modeling perspective}\label{sec:ref_pen}

In Section~\ref{sec:multivariate-est}, we have adopted two different approaches to produce smooth estimates $\hat \beta^G:\R^3 \ra \R^3$ and $\hat \beta^F:\M \ra \R$. The smoothness of $\hat \beta^F$ is enforced by means of a penalty $J(\cdot)$ defined in terms of a Sobolev norm, which implicitly defines a kernel $K_\M$. On the other hand, the smoothness of $\hat \beta^G$ is enforced by means of a norm $\| \cdot \|_{\calV(\R^3)}$, defined implicitly through the direct definition of a kernel $K_{\R^3}$. For clarity, we summarize the relevant function spaces and associated norms and kernels in Table~\ref{table:spaces}.

\begin{table}
\begin{center}
\begin{tabular}{ |c|c|c|c| }
 \hline
 Estimate & Function space & Norm & Kernel \\ 
  \hline
  $\hat \beta^F:\M \ra \R$ & $\calW^2(\M)$ & $J^{1/2}(\cdot) = \left( \|\Delta_{\M} \cdot\|^2_{\calL^2(\M)} + \epsilon \|\cdot\|^2_{\calL^2(\M)} \right)^{1/2}$ & $K_\M$ (implicit)\\ 
 $\hat \beta^G:\R^3 \ra \R^3$  & $\calV(\R^3)$ & $\| \cdot \|_{\calV(\R^3)}$ (implicit) & $K_{\R^3}$\\ 
 \hline
\end{tabular}
\caption{Table summarizing estimates and associated function spaces, norms and kernels.}\label{table:spaces}
\end{center}
\end{table}

From a methodological perspective, the reproducing kernel $K_{\R^3}(p,q)$ can be understood as a measure of the influence of the function value at $p \in \R^3$ on the function value at $q \in \R^3$ and vice-versa. Defining a smooth function space through a kernel has arguably an advantage when it comes to discretizing an infinite-dimensional minimization problem over that function space. In fact, thanks to the well-known representer theorem \parencite{wahba1990spline, yuan2010reproducing}, under mild conditions, its \textit{exact} solution can be expressed as a linear combination of the elements of a $n$-dimensional basis, which involves the explicit expression of the kernel. 

Hence, it is natural to wonder whether a similar approach could be adopted for $\hat \beta^F:\M \ra \R$. In other words, can we define a smooth real function space on $\M$ by explicitly defining a kernel $K_\M: \M \times \M \ra \R$ encoding a measure of influence that is coherent with the nonlinear geometry of $\M$? This, however, is a challenging task due to the positive-definiteness property that $K_\M$ must satisfy. Consider, for instance, the popular exponential kernel. Its natural extension to the manifold setting is $K_\M(p,q) = \exp(-c \, \text{d}_{\M}(p,q)^2)$, where $\text{d}_{\M}(p,q)$ is the geodesic distance between $p \in \M$ and $q \in \M$. Unfortunately, this kernel cannot be guaranteed to be positive definite for a general nonlinear manifold $\M$ \parencite{feragen2015geodesic,jayasumana2015kernel}. 

Alternatively, we could try to compute an explicit form of the kernel $K_\M$ from $J^{1/2}(\cdot)$ by employing the following identity \parencite{wahba1990spline, fasshauer2013reproducing}
\begin{equation}\label{eq:RKHS_eval}
f(p) = \langle K_\M(p, \cdot), f \rangle_{\calW^2(\M)}, \qquad \forall p \in \M, f \in \calW^2(\M),
\end{equation}
where $\langle \cdot, \cdot \rangle_{\calW^2(\M)}$ is the inner product that induces the norm $J^{1/2}(\cdot)$.
However, closed-form solutions to equation (\ref{eq:RKHS_eval}) are not available in our setting. Approximate solutions could be computed by Finite Elements Analysis (FEA) \parencite{quarteroni2009numerical}, but we would still face the challenge of storing the dense object $K_\M(\cdot,\cdot)$. As described in Section~\ref{sec:discretization}, we instead leverage FEA to directly discretize the function $\hat \beta^F:\M \ra \R$ in equation (\ref{eq:model_ls_biv}).

This highlights that the choice of the two modeling approaches is not arbitrary and that, arguably, for functional estimates supported on Euclidean spaces, defining explicitly a reproducing kernel is likely the preferred choice. Meanwhile, for general non-Euclidean domains, where defining a reproducing kernel is not trivial, the differential penalization approach is preferable.

\subsection{Theory}

Define the diagonal matrix of reproducing kernels $\vect K(p, q)$ with entries $\vect K_{11}(p_1, q_1) = K_{\R^3}(p_1, q_1)$ and $\vect K_{22}(p_2, q_2) = K_{\M}(p_2, q_2)$; $p_i \in D_i, p_j \in D_j$. Let $L_{\vect K}: \calH \ra \calH$ be the associated integral operator and, analogously to the univariate functional setting, define the sandwich operator $T = L^{1/2}_{\vect K} L_{\vect C} L^{1/2}_{\vect K}$. We make the following assumptions, which are analogous to Assumptions~\ref{assm:bounded}-\ref{assm:effective}.

\begin{assumption}\label{assm:bounded-mult}
The constant $\kappa^2_2$, defined as $\kappa^2_2 = \esssup \|L^{1/2}_{\vect K} \vect X\|^2_{\calH}$ is finite.
\end{assumption}

\begin{assumption}\label{assm:smooth-mult}
There exists a smooth function $\vect \beta^0 \in \calV(\R^3) \times \calW^2(\M)$ such that  $\vect \beta^0 = L_{\vect C}^{-1}(\vect \mu_2 - \vect \mu_1)$.
\end{assumption}

\begin{assumption}\label{assm:effective-mult}
The penalty coefficient $\lambda := \lambda_1 = \lambda_2 $  and the effective dimension of $T$ satisfy $D(\lambda) = \operatorname{Tr}((T + \lambda I)^{-1}T) \leq c \lambda^{-\theta}$ for constants $c,\theta > 0$.
\end{assumption}

The following theorem, which is an extension of Theorem \ref{thm:risk_manifold}, provides an upper bound for the out-of-sample risk.
\begin{theorem}\label{thm:risk_biv}
Under Assumptions~\ref{assm:bounded-mult}-\ref{assm:effective-mult}, if $\lambda \asymp n^{-\frac{1}{1+\theta}}$, the estimator $\vect{\hat \beta} = \left(\hat \beta^G, \hat \beta^F \right)$ in equation (\ref{eq:model_ls_biv}) is such that
\begin{align}
\E^{*} \left[ \langle \vect X^*, \vect \beta^0 - \vect{\hat \beta} \rangle_{\calH} \right] ^2 = \calO_p \left( n^{-\frac{1}{1+\theta}} \right),
\end{align}
where $\vect X^*$ is a copy of $\vect X$ that is independent of the training data and $\E^*$ is the expectation taken over $\vect X^*$.
\end{theorem}

\section{Discretization}\label{sec:discretization}

Consider a triangle mesh, denoted by $\mathcal{M}_\mathcal{T}$, which is formed by the union of a finite set of triangles, $\mathcal{T}$. These triangles share a common set of $s$ vertices, denoted as $\xi_1, \ldots, \xi_s$. Let $\mathcal{M}_\mathcal{T}$ be an approximate representation of the manifold $\mathcal{M}$. We then introduce the linear finite element space $\calW_\mathcal{T}$ consisting of a set of globally continuous functions over $\mathcal{M}_\mathcal{T}$ that are affine within each triangle $\tau$ in $\mathcal{T}$, i.e.,
\begin{equation*}
\calW_\mathcal{T} = \{ w \in C^0(\mathcal{M}_\mathcal{T}): w|_{\tau} \text{ is affine } \forall \tau \in \mathcal{T} \}.
\end{equation*}
The space $\calW_\mathcal{T}$ is spanned by the Finite Elements (FE) basis $\psi_1, \ldots, \psi_s$, where $\psi_l(\xi_j)=1$, if $l=j$, and $\psi_l(\xi_j)=0$ otherwise. In Figure~\ref{fig:basis}, we show one element of this basis. Moreover, define $\vect{\psi}$ as the vector-valued function $\vect{\psi}(\cdot) = \left(\psi_1(\cdot), \ldots, \psi_s(\cdot) \right)^T$. Our goal is to find an approximate solution $\hat \beta^F_\mathcal{T}$ of the form
\begin{equation}\label{eq:beta_f_discr}
\beta_\mathcal{T}^F(\cdot) = \sum_{l=1}^s c^F_l \psi_l(\cdot) =  (\vect{c}^{F})^T \vect{\psi}(\cdot),
\end{equation}
where $\vect{c}^{F} = \left(c^F_1, \cdots, c^F_s \right)^T$.

Let now $M$ and $S$ be the sparse mass and stiffness $s \times s$ matrices defined as $(M)_{jj'} =\int_{\mathcal{M}_{\mathcal{T}}} \psi_j \psi_{j'}$ and $(S)_{jj'}=\int_{\mathcal{M}_\mathcal{T}} \nabla_{\mathcal{M}_\mathcal{T}} \psi_j \cdot \nabla_{\mathcal{M}_\mathcal{T}} \psi_{j'}$, where $\nabla_{\mathcal{M}_\mathcal{T}}$ is the gradient operator on the mesh $\mathcal{M}_\mathcal{T}$. For $\beta_\mathcal{T}^F$ of the form given in equation (\ref{eq:beta_f_discr}), we have that the penalty term $J(\cdot)$ can be approximated as $(\vect c^F)^T D_{\mathcal{M}_\mathcal{T}} \vect c^F$, with $D_{\mathcal{M}_{\mathcal{T}}} = SM^{-1}S + \epsilon M$ \parencite{lila2016smooth}. Further, following an approach often adopted in FEA \parencite{fried1975finite,hinton1976note}, we replace the dense matrix $M^{-1}$ with the sparse matrix $\tilde{M}^{-1}$, where $\tilde M$ is the diagonal matrix such that $\tilde M_{jj} = \sum_{j'} M_{jj'}$. This results in the sparse penalty matrix $D_{\mathcal{M}_{\mathcal{T}}} = S\tilde{M}^{-1}S + \epsilon M$. In practice, each functional observation $x_i$ is also of the form given in equation (\ref{eq:beta_f_discr}). Therefore, denoting by $\mathbb{X}$ the $n \times s$ matrix where each row consists of the basis coefficients of $x_i$, the terms $\{\langle x_i, \beta^F \rangle_{\calL^2(\M)} \}$ can be approximated by the entries of the vector $\mathbb{X} M \vect c^F$.

We now turn our attention to the estimate $\hat \beta^G \in \calV(\R^3)$. 
Since for $\calV(\R^3)$ we have an explicit form of the associated reproducing kernel $K_{\R^3}$, we employ the representer theorem \parencite{wahba1990spline,yuan2010reproducing} and take $\hat \beta^G$ of the form 
\begin{equation}\label{eq:beta_g_discr}
\beta^G (\cdot) = \sum_{i=1}^n
c^G_i \int_{\R^3} K_{\R^3}(p, \cdot) v_i(p) \, dp.
\end{equation}
On the right hand side of Figure~\ref{fig:basis}, we show an example of a basis function $\int_{\R^3} K_{\R^3}(p, \cdot) v_i(p)\,dp$. Then, we have that $\|\beta^G\|^2_{\calV} =  (\vect{c}^G)^T \Sigma \vect{c}^G$, where 
$\vect{c}^G = \left(c_1^G, \ldots, c_n^G \right)^T$ and $\Sigma$ is a $n \times n$ matrix with entries
\[
\Sigma_{ij} =
\int \int v_i(p)^T K_{\R^3}(p, q) v_j(q)\, dp\, dq.
\]

\begin{figure}[!htb]
\centering
\includegraphics[width=0.6\textwidth]{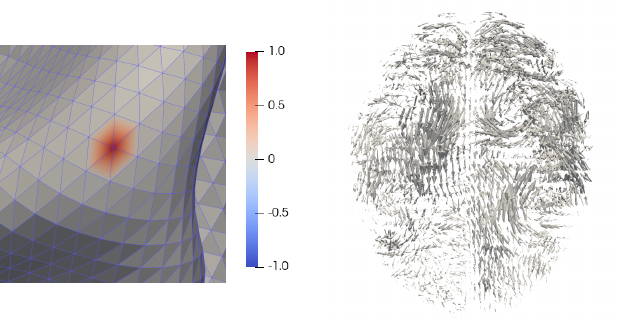}
\caption[]{On the left hand side, we show an element of the FE basis $\left\{\psi_l: \M_{\calT} \ra \R, l=1,\ldots,s \right\}$. This is a scalar affine function within each triangle of the mesh $\M_{\calT}$ that takes value $1$ on a fixed vertex and value $0$ on every other vertex. On the right hand side, we show an element of the basis $\left\{\int_{\R^3} K_{\R^3}(p, \cdot) v_i(p) \,dp, i=1,\ldots,n \right\}$. This is a smooth vector-valued function from $\R^3$ to $\R^3$.}
\label{fig:basis}
\end{figure}

As a result, the coefficients of the approximate solution of the model in equation (\ref{eq:model_ls_biv}) are given by
\begin{equation}
\left(\vect{\hat{c}}^G, \vect{\hat{c}}^F \right) = \argmin_{\vect c^G \in \R^{n}, \vect c^F \in \R^{s}} \left\|\vect y - \Sigma \vect c^G - \mathbb{X} M \vect c^F \right\|^2_2 + \lambda_1  (\vect c^G)^T \Sigma \vect c^G + \lambda_2 (\vect c^F)^T D_{\mathcal{M}_\mathcal{T}} \vect c^F,
\end{equation}
where $\vect y = \left(y_1, \ldots, y_n\right)^T$ is the vector of auxiliary response variables. It is easy to check that this minimization problem can be equivalently written as the following augmented quadratic least-squares problem
\begin{equation}\label{eq:linear_system}
\left(\vect{\hat{c}}^G, \vect{\hat{c}}^F \right) = \argmin_{\vect c^G \in \R^{n}, \vect c^F \in \R^{s}} \left\|\begin{bmatrix}
		\vect y\\
		\vect{0}
	\end{bmatrix} - A \begin{bmatrix}
		\vect{c}^G\\
		\vect{c}^F
	\end{bmatrix} \right\|^2_2,
\end{equation}
where $\vect{0}$ is the zero-vector of length $2s+n$ and with
\begin{equation*}
A =	\begin{bmatrix}
		\Sigma& \mathbb{X} M\\
		0& \lambda_2^{\frac{1}{2}} \tilde{M}^{-\frac{1}{2}} S \\
		0& \lambda_2^{\frac{1}{2}}\epsilon^{\frac{1}{2}} M^{\frac{1}{2}}\\

		\lambda_1^{\frac{1}{2}}\Sigma^{\frac{1}{2}} &0
	\end{bmatrix}.
\end{equation*}
Note that for $n \ll s$, which is the setting of our application, the matrix $A$ is sparse. Therefore, the minimization problem (\ref{eq:linear_system}) can be efficiently solved by conjugate gradients, or its variations, e.g. LSQR \parencite{paige1982algorithm}, without requiring the explicit computation of the high-dimensional normal matrix $A^T A$ -- a quantity related to the covariance structure of the functional predictors. 

An approximate solution to the univariate model in equation (\ref{eq:model_ls_univ}) follows as a special case of the multivariate case considered here.

\section{Application}\label{sec:application}
\subsection{Data and preprocessing}
We analyze a cohort of $n=484$ subjects from the ADNI and PPMI studies. On the basis of the ATN classification scheme \parencite{jack2016unbiased}, each subject was assigned to one of the two diagnostic categories -- C: Control ($n_1 = 100$) and AD: Alzheimer's Disease ($n_2 = 384$).
Here, we focus on data collected at the baseline visit, which includes, among other imaging modalities, structural T1-weighted MRI.  

The T1-weighted images were preprocessed using FreeSurfer \parencite{dale1999cortical, fischl1999cortical}. Specifically, white matter, grey matter, and cerebrospinal fluid were segmented and used to extract the outer and inner surfaces of the cerebral cortex. From these two surfaces, we generated a central surface interpolating the midpoints between the outer and inner layers, which offers an accurate representation of the two-dimensional anatomical structure of the cerebral cortex. This representation has the benefit of encoding a notion of distance between brain regions that is neurologically more relevant than the original volumetric representation. The cortical surface can moreover be coupled with one or more maps describing complementary structural or functional properties of the cortex, such as cortical thickness measurements \parencite{fischl2000measuring}, fMRI signals, or connectivity maps \parencite{smith2013functional,yeo2011organization}. In this study, we focus on cortical thickness, which is estimated from the distances between the outer and inner surfaces of the cerebral cortex. Next, the $n$ surfaces were registered and sub-sampled.

As a result of the preprocessing stage, we obtain $n = 484$ triangle meshes $\{\M^\calT_i\}$ consisting of $s=64$K vertices in correspondence across subjects, along with a set of triangular faces defining how these vertices are connected to delineate the surfaces. By classical Generalized Procrustes Analysis \parencite{dryden2016statistical}, translation, rigid rotation, and scale were removed from the data, while jointly estimating a template $\M^\calT$, which is also a triangle mesh with $s=64$K vertices in correspondence with those of the individual subjects. Each surface has been moreover coupled with cortical thickness measurements at the mesh vertices, which we model as a real piecewise linear continuous function $z^{\calT}_i$ over the mesh $\M^{\calT}_i$. 

The preprocessing stage results in a set of FoSs $\left\{\left(\M^{\calT}_i, z^{\calT}_i \right), i=1,\ldots,n \right\}$, which are discretized versions of the continuous idealized objects $\left\{\left(\M_i, z_i \right), i=1,\ldots,n \right\}$ introduced in Section~\ref{sec:lin_repr}. To simplify the notation, we drop the superscript $\calT$. Moreover, we denote the diagnostic labels by $\left\{g_i \in \{\text{C},\text{AD}\}, i=1,\ldots,n \right\}$. Three examples of such FoSs, and associated diagnostic labels, are given in Figure~\ref{fig:setting}A.

Here, we are interested in using the proposed models to identify subjects with AD from the extracted FoSs. The interpretability of these methods is an important feature. Indeed, while it is crucial to build models with good classification accuracy, it is equally important to describe the estimated relationship between the predictors and the outcome variable, in order to inform subsequent studies and generate data-driven hypotheses about the pathophysiology of AD.

\subsection{Linear functional representations}
For each FoS, we compute a function $v_i \in \calV(\R^3)$ such that $\varphi_{v_i}(\M)$ closely approximates $\M_i$, where closeness is measured as the sum of Euclidean distances between the corresponding vertices of $\varphi_{v_i} \left(\M \right)$ and $\M_i$. As noted in Appendix~\ref{sec:diff_op}, alternative definitions of surface similarity could also be used. We can then transport the function $z_i: \M_i \ra \R$ onto the template defining a continuous piecewise linear function $x_i = z_i \circ \varphi_{v_i}$. This leads to the definition of the bivariate functional representation $\left(v_i, x_i \right)$ that is a linear representation of the FoS $(\M_i , z_i) \approx (\varphi_{v_i} (\M) , x_i \circ \varphi^{-1}_{v_i})$. Further details on the computation of $v_i$ can be found in Appendix~\ref{sec:diff_op}.

\subsection{Discriminant analysis}
Our aim is to estimate directions in the parametrizing geometric and thickness spaces that are most predictive of AD. To this end, we apply the model introduced in Section~\ref{sec:multivariate-est} to the training data $\{g_i, v_i-\bar{v}, x_i-\bar{x}\}$, where $\bar{v}$ and $\bar{x}$ are the sample means of $\{v_i\}$ and $\{x_i\}$. The training data are a subset of the entire dataset containing 50\% of the observations.
From this process, we derive the estimates $\hat \beta^G: \R^3 \ra \R^3$ and $\hat \beta^F: \M \ra \R$. Given a new subject with predictors $(v^*,x^*)$, these estimates can be used in conjunction with the classification rule $\langle v^* -\bar{v}, \hat \beta^G \rangle + \langle x^* - \bar{x}, \hat \beta^F \rangle > c^{\text{th}}$ to predict the diagnostic label of a new subject. The cut-off level $c^{\text{th}}$ can be chosen by computing sensitivity and specificity on a test set, for different values $c^{\text{th}}$, and by selecting the desired level and type of accuracy.

\begin{figure}[!htb]
\centering
\includegraphics[width=1\textwidth]{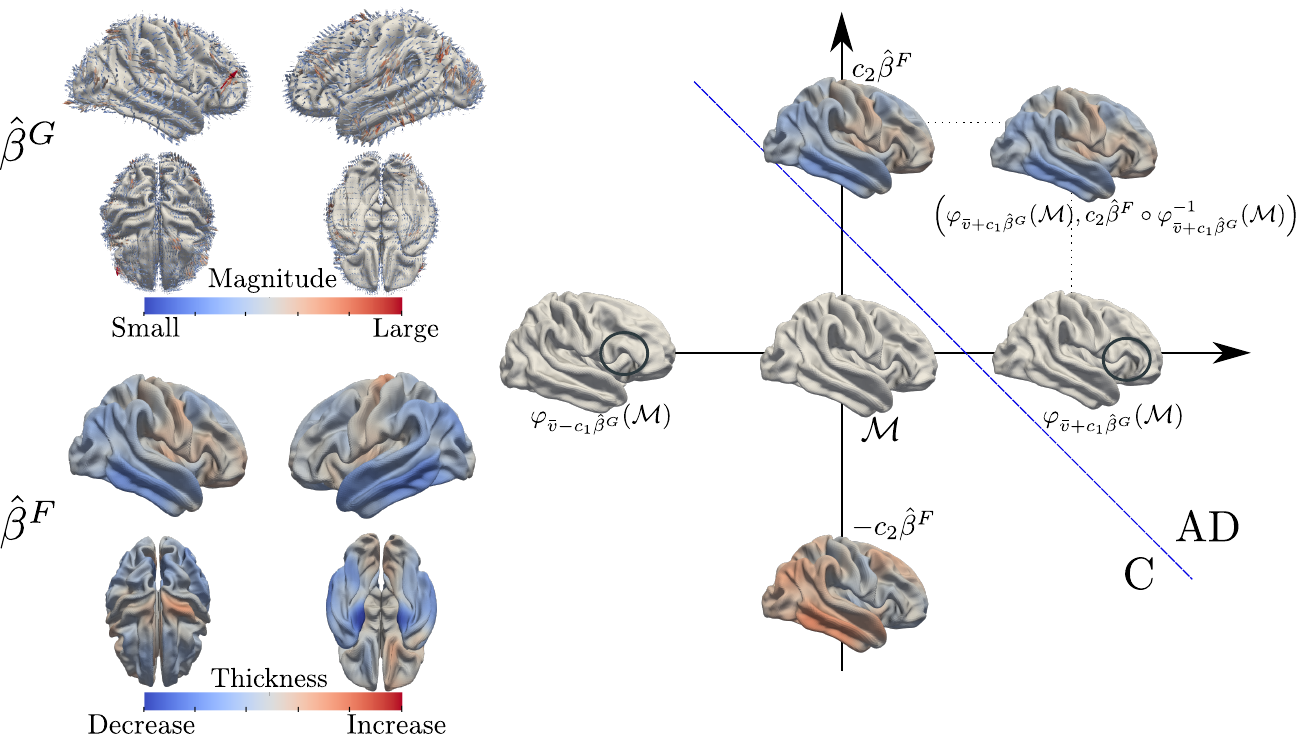}
\caption[]{On the left hand side, we show the most discriminant geometric and thickness directions as estimated from the linear representations $\{(v_i-\bar{v}, x_i-\bar{x})\}$. These are a vector field $\hat \beta^G: \R^3 \ra \R^3$, representing the most predictive geometric pattern of AD, and a function $\hat \beta^F: \M \ra \R$, representing the most predictive cortical thickness pattern of AD. For a new FoS, with linear representation $(v^*, x^*)$, we compute the score $\langle v^* -\bar{v}, \hat \beta^G \rangle + \langle x^* -\bar{x}, \hat \beta^F \rangle$ and predict whether the subject has AD by comparing the score value with a predetermined threshold $c^{\text{th}}$. On the right hand side, we depict the process of mapping back the estimates $\hat \beta^G$ and $\hat \beta^F$ to the space of FoSs. On the same space, we also pictorially map the classification rule adopted. In the $\hat \beta^F$ figure, the blue regions represent the areas of the cortical surface where a thinner cortex, relative to the population average, is indicative of AD. These are mostly localized in the lateral temporal, entorhinal, inferior parietal, precuneus, and posterior cingulate cortices. The red arrows in the $\hat \beta^G$ figure represent the regions where differences in the morphological configuration of the cerebral cortex, compared to the population average, are most predictive of AD. The specific types of morphological changes can be inspected by comparing the surfaces $\varphi_{\bar{v} - c_1 \hat \beta^G}(\M)$ and $\varphi_{\bar{v} + c_1 \hat \beta^G}(\M)$, on the right hand side diagram. 
}
\label{fig:results}
\end{figure}

These estimates effectively identify linear directions $\{c_1 \hat \beta^G, c_1 \in \R \}$ and $\{c_2 \hat \beta^F, c_2 \in \R \}$, in their respective spaces, that can be interpreted as the most discriminant geometric and thickness directions. Specifically, large values of $c_1 \in \R$ and $c_2 \in \R$ describe configurations $c_1 \hat \beta^G$ and $c_2 \hat \beta^F$ that are predictive of AD. Low values of $c_1 \in \R$ and $c_2 \in \R$ describe configurations that are instead predictive of the subject being healthy. Moreover, given the additive modeling assumption on the geometric and thickness components, for every configuration $\left(c_1 \hat \beta^G, c_2 \hat \beta^F \right)$, an increase of $c_1$ ($c_2$) for a fixed $c_2$ ($c_1$) describes a configuration that is more strongly associated with AD.  

Crucially, these linear trajectories on the parametrizing space can be mapped back to the original space of FoSs by using equation (\ref{eq:FoS_repr}), defining the curved space
\[
\left( \varphi_{\bar{v} + c_1 \hat \beta^G}(\M), c_2 \hat \beta^F \circ \varphi^{-1}_{\bar{v} + c_1 \hat \beta^G}\right), \qquad c_1, c_2 \in \R.
\]

We fit the proposed model for different choices of the parameters $\lambda_1$ and $\lambda_2$. Recall that $\lambda_1$ controls the regularity of the geometric discriminant direction and $\lambda_2$ that of the thickness discriminant direction, with high values virtually constraining the solution to be the zero function. The final choice of $\lambda_1$ and $\lambda_2$ is the result of a compromise between classification accuracy on the test set and the consistency of the estimated discriminant directions with the neurodegenerative nature of the disease (see also Discussion). The test Area under the ROC Curve (AUC) of the selected model is 0.7006.

\subsection{Results}\label{sec:results}
On the left hand side of Figure~\ref{fig:results}, we show the estimated most discriminant geometric and thickness directions, i.e., $\hat \beta^G: \R^3 \ra \R^3$ and $\hat \beta^F: \M \ra \R$. These have been estimated by applying the model in equation (\ref{eq:model_ls_biv}) to the linear representations $\{(v_i-\bar{v}, x_i-\bar{x})\}$. The colormap describing $\hat \beta^F$ illustrates what types of variations, with respect to the population average cortical thickness, are most predictive of AD. Specifically, a thinner cerebral cortex in the blue areas (i.e., lateral temporal,  entorhinal, inferior parietal, precuneus, and posterior cingulate cortices) is associated with AD. These results are consistent with the typical thickness signature of AD observed to date \parencite[see, e.g., ][]{bondareff1989neurofibrillary,dickerson2009cortical,sabuncu2011dynamics}. The geometric component $\hat \beta^G$ is instead a vector field in $\R^3$. This is a linear representation of the morphological variations, with respect to the population average cortex geometry, that are associated with AD. While a full understanding of its meaning is only possible by mapping $\hat \beta^G$ back to the space of FoSs, i.e., by examining $\varphi_{\bar{v} + c_1 \hat \beta^G}(\M)$ for different choices of $c_1 \in \R$, the magnitude of the vector field $\hat \beta^G$, at any fixed point, offers a rough indication of the cortical regions whose morphological variations are most relevant to the classification problem.

On the right hand side of Figure~\ref{fig:results}, we show the FoSs associated with the linear representations $\hat \beta^G$ and $\hat \beta^F$, that is, $\left( \varphi_{\bar{v} + c_1 \hat \beta^G}(\M), c_2 \hat \beta^F \circ \varphi^{-1}_{\bar{v} + c_1 \hat \beta^G}\right)$ with $c_1, c_2 \in \R$. These describe the most predictive patterns of AD in terms of the original neurobiological objects. We have circled a specific area of the brain to ease comparison and highlight the morphological patterns that the model deems relevant to the classification problem. 

\subsection{Comparison against alternative approaches}\label{sec:appl_comparison}

In this section, we compare the test AUC of our proposed classification method with alternative models and evaluate different representation models for FoSs.  In addition to the functional linear discriminant model (FLDA) that we propose, we also consider the following alternatives: (i) FPCA+LDA: The geometry-aware FPCA model proposed in \textcite{lila2016smooth}, followed by multivariate LDA \parencite{hastie2009elements} on the PC scores; (ii) Lasso: A logistic regression model with lasso regularization \parencite{tibshirani1996regression}; (iii) Ridge: A logistic regression model with an $\ell^2$ regularization \parencite{hoerl1970ridge}; (iv) FQDA: The approximate functional quadratic discriminant model defined in Section~\ref{sec:non_linear}; (v) RF: A Random forest model \parencite{breiman2001random}; (vi) SVM: A support vector machine with a squared exponential kernel \parencite{cortes1995supportvector}; (vii) NN: A multilayer feedforward neural network \parencite{hastie2009elements}.

\begin{table}[h!]
\scriptsize
\centering
 \begin{tabular}{l c c c c c c c c} 
 \toprule
 Representation model & \multicolumn{4}{c}{Linear methods} & \multicolumn{4}{c}{Nonlinear methods}\\
 \midrule
  & FLDA & FPCA+LDA & Lasso & Ridge & FQDA & RF & SVM & NN \\ [0.5ex] 
  \cmidrule(lr){2-5}\cmidrule(lr){6-9}
 Thickness & 0.7626 & 0.7583  &  0.7487  &  0.7632 & \textbf{0.7710} &  0.6043  &  0.7597  &  0.7678 \\ 
 Thickness \& Displacement & 0.6623 & - & 0.6626  &  0.6571 & - & 0.6742 & \textbf{0.6861}  &  0.6771 \\
 Thickness \& Shape spectrum & - & - & \textbf{0.7832}  &  0.6638 & - &   0.5797  &  0.7606  &  0.6878\\
 Proposed FoSs representation & \textbf{0.7716} & - & 0.7484  &  0.7646 & - &   0.7132  &  0.7600 &  0.7443\\
 \bottomrule
 \end{tabular}
 \caption{\label{table:AUC_final} The test AUC of the classification methods applied to the data of our final application. Four different representation models have been considered: (i) the registered thickness map $x_i:\M \ra \R$ without geometric information; (ii) the parametrization $h_i: \M \ra \R^4$, where the first three components are the surface coordinates, and the last component is thickness; (iii) the registered thickness map $x_i:\M \ra \R$ and the first 200 eigenvalues of the Laplace-Beltrami operator computed on the surface $\M_i$; and (iv) the proposed representation model $(x_i, v_i)$. The symbol `-' indicates that although the method could be adapted to accommodate the specific FoS representation model, its implementation is beyond the scope of this paper and is left to future work. For each representation model, the top-performing method is highlighted.}
\end{table}

Furthermore, besides the proposed representation model $(v_i,x_i)$ for FoSs, we also consider the following representations: (i) Thickness: Spatially normalized thickness maps $x_i:\M \ra \R$ without geometric information; (ii) Thickness \& Displacement: The parametrizations $h_i: \M \ra \R^4$, where the first three components are the surface coordinates, and the last component is the (spatially normalized) thickness map; (iii) Thickness \& Shape spectrum: The spatially normalized thickness maps $x_i:\M \ra \R$ and the first 200 eigenvalues of the Laplace-Beltrami operator computed on the surface $\M_i$, i.e., a spectral representation of shape \parencite{reuter2006laplacebeltrami}.

To evaluate the listed methods and representation models, we split the dataset into three sets, namely the training set, validation set, and test set, comprising $50\%$, $20\%$, and $30\%$ of the data, respectively. While a Monte Carlo evaluation of these methods would be desirable, it is computationally prohibitive, so we defer that analysis to the simulation setting in Appendix~\ref{sec:simulations}. However, we use the same exact data split for all methods. The models are trained on the training set, hyperparameters are chosen to maximize the AUC on the validation set and the selected model is tested on the test set, resulting in the AUC scores presented in Table~\ref{table:AUC_final}. Note that, in contrast to the results presented in Section~\ref{sec:results} and Figure~\ref{fig:results}, all hyperparameters of the proposed methods have been chosen to maximize the AUC on the validation set, rather than striking a balance between the classification accuracy and consistency of the estimated discriminant directions with the neurodegenerative nature of the disease. Hence, the test AUC value of the proposed method is different from that in Section~\ref{sec:results}.

\begin{figure}[!htb]
\centering
\includegraphics[width=1\textwidth]{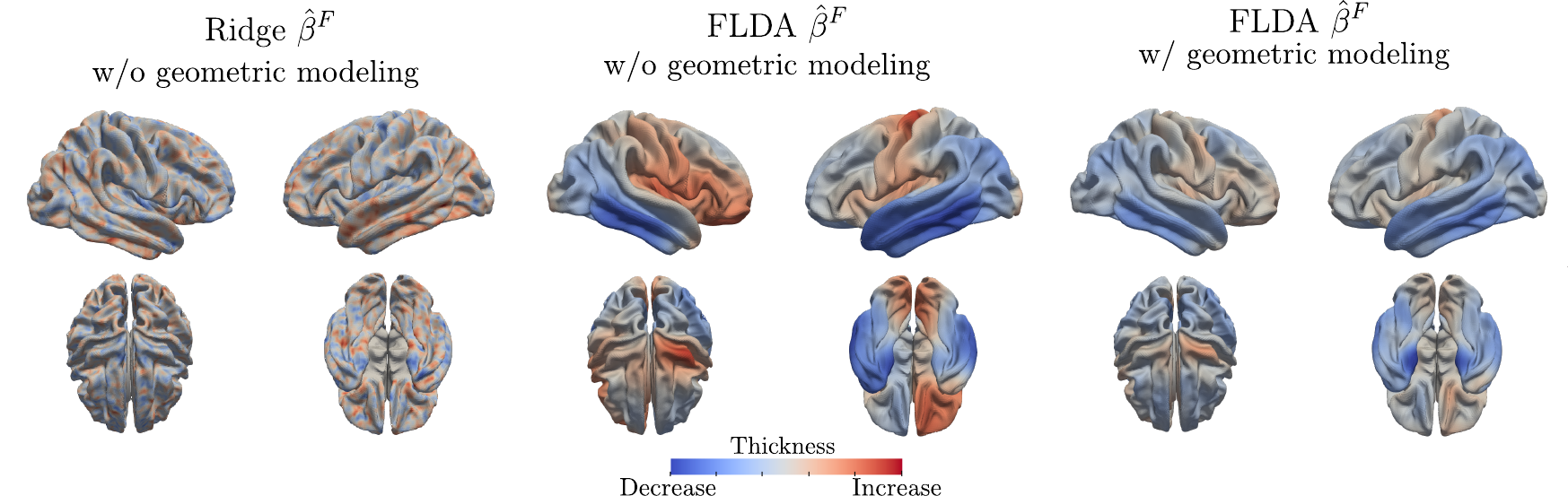}
\caption[]{On the left side, we show the discriminant direction derived from applying a ridge logistic regression model to the thickness maps. In the center, we show the discriminant direction resulting from fitting the proposed model in equation (\ref{eq:model_ls_univ}) to the thickness maps. Although it does not account for subject-specific geometric variations, this model enforces smoothness. On the right side, we have the cortical thickness discriminant direction obtained by fitting the model in equation (\ref{eq:model_ls_biv}), which explicitly accounts for inter-subject geometric differences. The results of the logistic regression are more difficult to interpret due to the high spatial variability. The model in equation (\ref{eq:model_ls_univ}) provides more interpretable results thanks to its smoothness penalty, but suggests that a \textit{thicker} cortex in the red areas is indicative of AD, which is not physiologically plausible. When we explicitly model geometric differences, this evidence seems to disappear. This suggests that there is a non-negligible dependence structure between the predictors modeling geometry and those modeling thickness. Differences that seemed to be related to cortical thickness in the model without the geometric component are now captured by the term that models cortical geometric variations. Furthermore, when we model inter-subject geometric differences the entorhinal cortex atrophy in the medial temporal lobe is identified as the strongest predictor of AD. This is consistent with pathological findings and staging of early AD \parencite{braak2006staging}.}
\label{fig:compare}
\end{figure}

For standard multivariate models, we use the values of $x_i$ at the vertices of the template mesh (64K values) and the RKHS coefficients of the estimated $v_i$ (192K coefficients) to construct the data matrix.  We have also implemented a variation of the functional linear discriminant model introduced in Section~\ref{sec:multivariate}, for multivariate functions whose components share a non-linear domain $\M$, in order to accommodate the representation $h_i: \M \rightarrow \R^4$.

The results are shown in Table~\ref{table:AUC_final}, from which we can make several observations. Firstly, if the goal is to maximize prediction accuracy, then the best-performing model is a lasso-penalized generalized linear model applied to thickness maps and the first 200 eigenvalues of the Laplace-Beltrami operator of the surfaces. However, as mentioned in the introduction, this ``lossy'' shape representation cannot be mapped back to the original space of neurobiological objects, leading to a less interpretable model. Additionally, our results show that in the context of our application, the representation $(v_i,x_i)$ performs better across all methods than using the representation $h_i: \M \ra \R^4$. The latter appears to be more susceptible to overfitting, resulting in inferior performance even when compared to models that use thickness only. Finally, classification using thickness alone produces satisfactory results. The top-performing models are the proposed FLDA and FQDA, and the ridge logistic regression model. One possible explanation for this is that the registered thickness maps may include some geometric information due to misregistration. While incorporating geometric information into the model may lead to only minor improvements in classification performance, as shown in Figure~\ref{fig:compare}, the estimated discriminant direction can be significantly different between the two models. Although the ground truth is unknown, the estimated discriminant direction when geometric information is included is more consistent with the neurodegenerative nature of the disease, as explained in the next section.

\subsection{Discussion}\label{sec:discussion}

The results in Figure~\ref{fig:results} identify the typical AD thickness signature. Several studies that focus on identifying AD-vulnerable areas include the regions found in our analysis \parencite[see, e.g., ][]{bondareff1989neurofibrillary,dickerson2009cortical,sabuncu2011dynamics}. However, there is some variability in the estimated regions. For instance, \textcite{sabuncu2011dynamics} used a dynamic model and found strongest changes in the inferior parietal regions and the posterior cingulate. It should be noted that these studies typically consist of massive univariate analyses between the cortical thickness at each vertex, or each parcel, and the diagnostic label. They are therefore taking a feature-centric perspective on the problem. It is not clear how these findings would generalize to out-of-sample data \parencite{li2020statistical}. 

To demonstrate the importance of modeling cortical geometry, we compare our results to those obtained by fitting a ridge logistic regression model and the proposed model in equation (\ref{eq:model_ls_univ}), i.e., by discarding inter-subject geometric differences. We compare these estimates in Figure~\ref{fig:compare}. What we observe is that ridge logistic regression yields estimated discriminant directions that are more difficult to interpret, due to the high spatial variability. Except for the entorhinal cortex, the functional model in equation (\ref{eq:model_ls_univ}) is able to capture the main areas where cortical thinning is associated with AD. However, this model also suggests that a thicker cortex in certain regions (dark red) is associated with AD, contradicting the neurodegenerative nature of AD. Interestingly, introducing the geometric component in the model reduces such effects. This may also be caused by the geometric component now capturing systematic misregistration. In order to verify such a hypothesis, further validation of the estimated geometric component is required in controlled settings where registration is more reliable, e.g., in the longitudinal setting.

\section{Conclusions}\label{sec:conclusion}

We introduce a framework for the discriminant analysis of functional data supported on random manifold domains, i.e., FoSs. To this aim, we adopt linear representations of these objects that are bivariate functional data belonging to linear spaces. We then define a functional linear classification model on the parametrizing space. Thanks to a penalized least-squares formulation, the proposed model is able to estimate the most discriminant direction in the data without requiring the explicit computation of the covariance function of the predictors or low-rank approximations thereof. This allows us to reduce the memory requirements by five orders of magnitude and ultimately be able to run our model on a standard workstation. The complexity of the solution is controlled by means of differential penalties that are aware of the geometry of the domain where the functional data are supported.

We apply the proposed model to the analysis of modern multi-modal neuroimaging data. Specifically, we estimate interpretable discriminant directions that are able to leverage both geometric and thickness features of the cerebral cortex to identify subjects with AD. Our results are consistent with those in the neuroscience literature. 

The model proposed can be applied to several imaging settings that lead to FoSs representations, such as musculoskeletal imaging \parencite{gee2018how} or cardiac imaging \parencite{biffi2018threedimensional}. It is also important to highlight that the proposed model is not a mere generalization of existing models for functional data supported on one-dimensional domains to multidimensional domains. We believe that its application to one-dimensional functional data, where the bivariate representation is given by the registered functions and associated registration maps, leads to a novel classification approach in this simplified setting.

\section*{Acknowledgments}
{\footnotesize
\noindent 
Data used in the preparation of this article were obtained from two sources: the Alzheimer's Disease Neuroimaging Initiative (ADNI) and the Parkinson's Progression Markers Initiative (PPMI). ADNI is funded by the National Institute on Aging, the National Institute of Biomedical Imaging and Bioengineering, and through generous contributions from the following: AbbVie, Alzheimer's Association; Alzheimer's Drug Discovery Foundation; Araclon Biotech; BioClinica, Inc.; Biogen; Bristol-Myers Squibb Company; CereSpir, Inc.; Cogstate; Eisai Inc.; Elan Pharmaceuticals, Inc.; Eli Lilly and Company; EuroImmun; F. Hoffmann-La Roche Ltd and its affiliated company Genentech, Inc.; Fujirebio; GE Healthcare; IXICO Ltd.; Janssen Alzheimer Immunotherapy Research \& Development, LLC.; Johnson \& Johnson Pharmaceutical Research \& Development LLC.; Lumosity; Lundbeck; Merck \& Co., Inc.; Meso Scale Diagnostics, LLC.; NeuroRx Research; Neurotrack Technologies; Novartis Pharmaceuticals Corporation; Pfizer Inc.; Piramal Imaging; Servier; Takeda Pharmaceutical Company; and Transition Therapeutics. The Canadian Institutes of Health Research is providing funds to support ADNI clinical sites in Canada. Private sector contributions are facilitated by the Foundation for the National Institutes of Health (\url{www.fnih.org}). The grantee organization is the Northern California Institute for Research and Education, and the study is coordinated by the Alzheimer's Therapeutic Research Institute at the University of Southern California. ADNI data are disseminated by the Laboratory for Neuro Imaging at the University of Southern California. 

PPMI --- a public-private partnership --- is funded by the Michael J. Fox Foundation for Parkinson's Research and funding partners. 
The complete list of PPMI funding partners can be found at \url{www.ppmi-info.org}.
 
ADNI data are available to the scientific community thorough the LONI Image and Data Archive at \url{http://adni.loni.usc.edu/data-samples/accessdata}. PPMI data can be accessed through \url{https://www.ppmi-info.org/access-data-specimens/download-data}.
}

\appendix
\appendixpage

\section{Linear functional representation model}
\subsection{Diffeomorphic deformation operator}\label{sec:diff_op}
The diffeomorphic operator $\varphi$ can be constructed as follows. Let $\{v_t \in \calV(\R^3): t \in [0,1]\}$ be a time-variant vector field such that $\int_0^1 \| v_t \|^2_{\calV(\R^3)}\, dt< \infty$. Then, the solution $\phi_v:[0,1] \times \R^3 \ra \R^3$, at time $t=1$, to the Ordinary Differential Equation (ODE)
\begin{align}\label{eq:ODE}
\begin{cases}
\frac{\partial \phi_v}{\partial t}(t,x) = v_t \circ \phi_v(t,x)\ \qquad &t \in [0,1], x \in \R^3,\\
\phi_v(0,x) = x \qquad &x \in \R^3,\\
\end{cases}
\end{align}
is a diffeomorphic deformation of $\R^3$ \parencite[see, e.g.,][]{younes2019shapes}. We then model $\{v_t:t \in [0,1]\}$ as a minimizer of the quantity $\int_0^1 \|v_t\|^2_\calV \, dt$, for a given initial vector field $v_0 \in \calV(\R^3)$ \parencite{miller2006geodesic}. Finally, the diffeomorphic operator is defined as $\varphi_{v_0}(x) = \phi_v(1,x)$, where $v_0 \in \calV(\R^3)$ is the initial vector field generating $\{v_t:t \in [0,1]\}$, and $\phi_v$ is the solution to the ODE in equation (\ref{eq:ODE}) for the computed $\{v_t:t \in [0,1]\}$.  

\subsection{Computation}\label{appendix:computation}
In practice, each surface $\M_i$ has a computational representation $\M^\calT_i$ that is a triangle mesh with $s$ vertices $\xi^i_1,\ldots,\xi^i_s$ in correspondence across the $n$ subjects. We use these vertices to perform Procrustes analysis \parencite{dryden2016statistical} and remove translation, size, and rigid rotations from the surfaces $\M^\calT_i$. In addition, Procrustes analysis yields a template $\M^\calT$ with vertices $\xi_1,\ldots,\xi_s$. 

The representation functions $\{v_i: v_i \in \calV(\R^3)\}$, associated with the surfaces $\{\M^\calT_i\}$, are then computed by solving the minimization problem
\begin{equation}\label{eq:minimization_geo}
{v}_i = \argmin_{v \in \calV(\R^3)} \sum_{l=1}^s \left\| \varphi_{v}(\xi_l) - \xi^i_l \right\|^2_{\mathbb{R}^3} + \lambda \|v\|^2_{\calV(\R^3)}, \qquad i=1,\ldots,n,
\end{equation}
where the least-squares term ensures that the deformed template $\varphi_{v_i}(\M^\calT)$ is a close approximation of $\M^\calT_i$. The term $\|v\|^2_{\calV(\R^3)}$ is a regularizing term that encourages the solution to achieve such a close approximation with a `minimal' deformation. The constant $\lambda$ is selected by inspecting the solutions on a small subset of the full cohort. To perform the actual computations, we use the MATLAB implementation fshapetk \parencite{charlier2015matching,charlier2017fshape}. Note that if the vertices were not in correspondence, that is, the surfaces had not been registered beforehand, the proposed framework would still be applicable by replacing the least-squares term in equation (\ref{eq:minimization_geo}) with a more general shape similarity measure $D(\cdot,\cdot)$. An example of such a similarity measure is found in \textcite{vaillant2005surface,vaillant2007diffeomorphic}, where the authors use the concept of currents, from geometric measure theory, to represent surfaces.

\subsection{Selecting an appropriate kernel}
The main requirement for choosing the kernel $K_{\R^3}$ is that the associated space $\calV(\R^3)$ is an admissible space \parencite{younes2019shapes}. Therefore, there must exist a positive constant $M$ such that for all $v \in \calV(\R^3)$, the following inequality holds:
\[
\|v\|_{1, \infty} \leq M\|v\|_{\calV(\R^3)},
\]
where $\| \cdot \|_{1, \infty}$ is the canonical norm of the space $C^1(\mathbb{R}^3)$. This condition guarantees that $\varphi_v$ is diffeomorphic for any $v \in \calV(\R^3)$.

Nonetheless, within the set of admissible spaces, the specific choice of the kernel can significantly influence the quality of the estimated representations. We have found the approach proposed in \textcite{bruveris2012mixture} to be effective. This uses a mixture of isotropic Gaussian kernels with different variances, which intuitively allows for a multi-scale representation of the surfaces. In our application, we use a mixture of six Gaussian kernels with variance parameters set to $(\sigma_1^2,\ldots,\sigma_6^2) = (64, 16, 4, 1, 0.25, 0.01)$. Our choice was informed by visual inspection of the differences between the estimated $\varphi_{v_i}(\M^\calT)$ and $\M_i^\calT$.

\subsection{Template estimation}\label{sec:template}

\RestyleAlgo{boxruled}
\begin{algorithm}
\KwData{The surfaces $\M_1,\ldots,\M_n$ and an initial guess for the template $\hat{\M}^{\{1\}} = \M_1$}
\KwResult{$\hat{\M} = \hat{\M}^{\{N_\text{iter}\}}$}
\For{\upshape $\text{iter} = 1,\ldots, N_\text{iter}$}{
\For{$i = 1,\ldots, n$}{
$\hat {v}^{\text{\{iter\}}}_i \gets \argmin_{v \in \calV(\R^3)} 	
	D\left( \varphi_{v}(\M^{\{\text{iter}\}} ), \M_i \right) + \lambda 	\|v\|^2_{\calV(\R^3)}$
	}
$\hat{\mu}^{\text{\{iter\}}} = \frac{1}{n} \sum_{i=1}^n \hat{v}^{\text{\{iter\}}}_i$

$\hat{\M}^{\text{\{iter\}}} = \varphi_{\hat{\mu}^{\text{\{iter\}}}}\left(\M \right)$
}
\caption{Algorithm for template estimation.}\label{alg:template}
\end{algorithm}

In this section, we introduce an algorithm designed to estimate a template leveraging the (formal) Riemannian structure of the manifold of diffeomorphisms adopted to model random manifold domains. One approach is to define the template so that the average of the linear representations $\{v_i\}$, located on the tangent space at the identity map, is zero. The details of this iterative centroid approach are outlined in Algorithm~\ref{alg:template}. For an overview of alternative approaches see \textcite{cury2014diffeomorphic}.

Despite leveraging GPU acceleration for computing RKHS norms and associated gradients, the process of computing $v_i$ for each subject still takes about 40 minutes in our application. This makes the process of estimating the template computationally prohibitive given the necessity of multiple iterations. Therefore, we have chosen to use a fixed template in our final application.

\section{Simulations}\label{sec:simulations}
In this section, we conduct simulations to assess the finite sample classification performance of the model proposed when compared to the models introduced in Section~\ref{sec:appl_comparison}. Here we focus on the functional univariate setting described in Section~\ref{sec:univariate}.

\begin{figure}[!htb]
\centering
\includegraphics[scale=0.8]{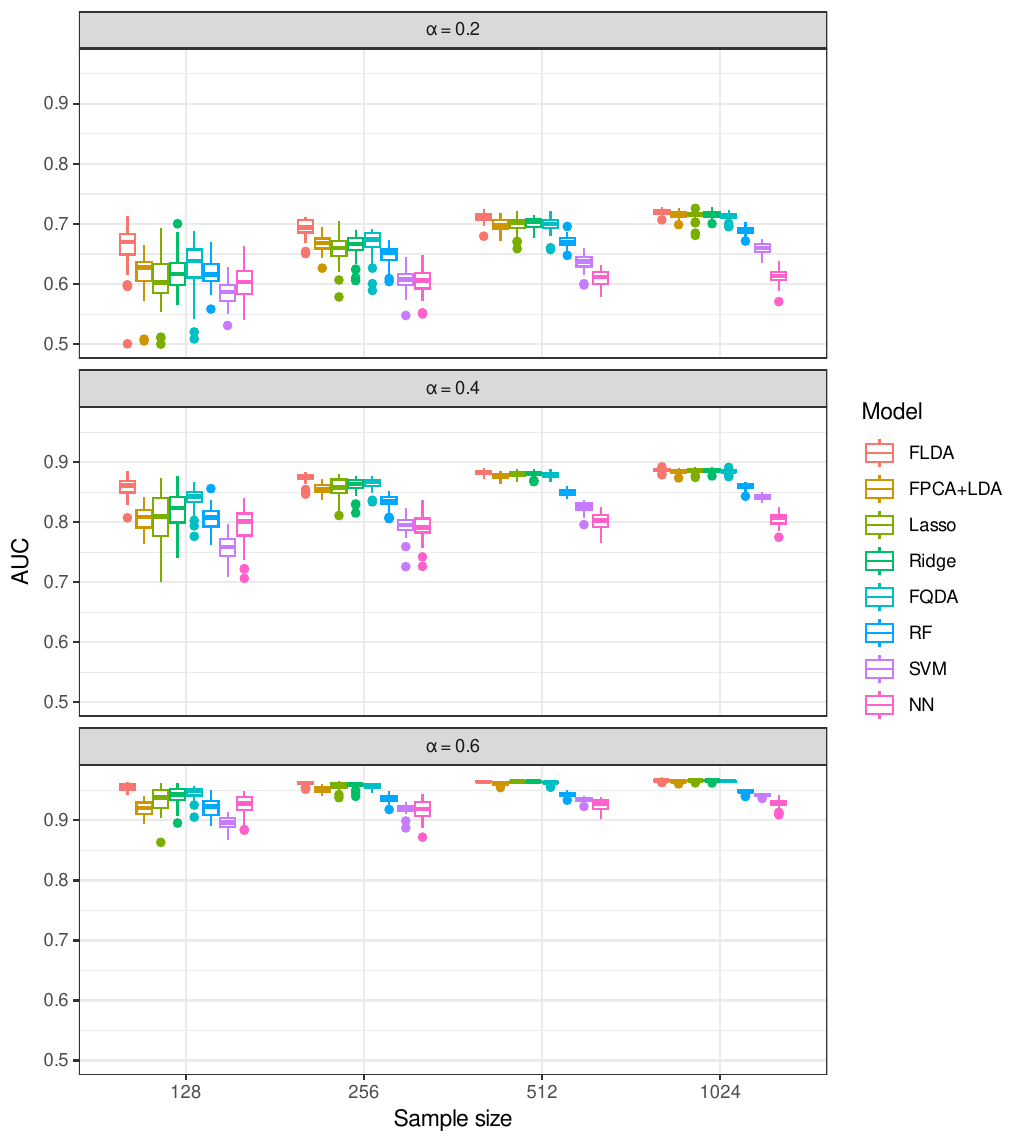}
\caption{Results of the simulation study to assess the performance of our proposed method, under the assumption of homogeneous covariances, for various sample sizes ($n=128,256,512,1024$) and signal-to-noise ratios ($\alpha=0.2,0.4,0.6$), where $\alpha$ reflects the strength of the discriminant signal. Prediction accuracy is measured using AUC and the simulations were repeated 50 times for each setting.}
\label{fig:AUC_homo}
\end{figure}

We use a triangle mesh $\mathcal{M}_{\mathcal{T}}$ with $642$ nodes that is an approximation of a brainstem. On this triangulated surface, we generate the orthonormal functions $\{v_l:l=1,2,\ldots,40\}$ consisting of $40$ eigenfunctions of the Laplace-Beltrami operator computed on $\mathcal{M}_{\mathcal{T}}$. Then, we generate two sets of smooth functional data supported on $\mathcal{M}_{\mathcal{T}}$, with identical within-group covariance structures, as follows:
\begin{equation}\label{eq:sim_generation}
\begin{aligned}
&\text{Group 1:} \quad
x_{i1}=w_{i 1} v_{1}+w_{i 2} v_{2}+\ldots+ w_{i 40} v_{40} \quad i=1, \ldots, n, \\
&\text{Group 2:} \quad
x_{i2}=\alpha\mu+u_{i 1} v_{1}+u_{i 2} v_{2}+\ldots+ u_{i 40} v_{40} \quad i=1, \ldots, n,
\end{aligned}
\end{equation}
where $u_{ij}$ and $w_{ij}$ are zero-mean independent random variables that represent the scores and are distributed according to a normal distribution with variance $\sigma_j^2$ decreasing in $j$. The function $\mu$ is the groups' difference, chosen to be a fixed linear combination of the eigenfunctions $\{v_l:l=1,2,\ldots,40\}$, and its magnitude is controlled by a parameter $\alpha > 0$ defining the `difficulty' of the classification problem, that is, the signal-to-noise ratio. We then identify three regimes, with $\alpha \in \{0.2,0.4,0.6\}$. Given a new function $x^*$, the aim is to recover the group this observation belongs to. 

\begin{figure}[!htb]
\centering
\includegraphics[scale=0.8]{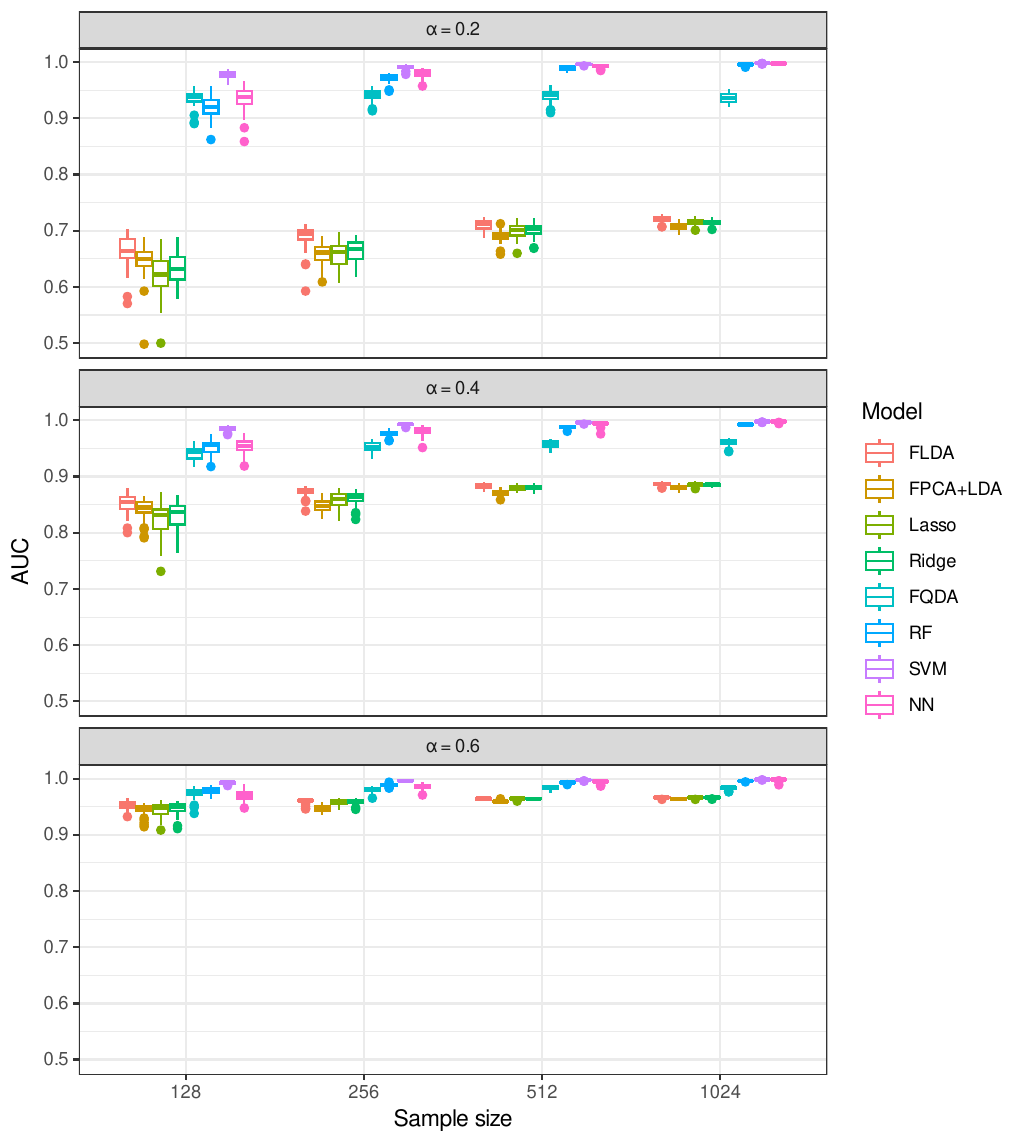}
\caption{Results of the simulation study for heterogeneous covariance structures across different sample sizes ($n=128,256,512,1024$) and signal-to-noise ratios ($\alpha=0.2,0.4,0.6$), where $\alpha$ reflects the strength of the discriminant signal. The prediction accuracy was evaluated through AUC and the simulations were repeated 50 times for each setting.}
\label{fig:AUC_het}
\end{figure}

We compare the proposed FLDA method against the models introduced in Section~\ref{sec:appl_comparison}, e.g., FPCA followed by LDA, Lasso and Ridge logistic regression, random forests, support vector machines, and fully-connected feed-forward neural networks. For standard multivariate models, we use the values of $x_i$ at the vertices of $\mathcal{M}_{\mathcal{T}}$ to construct the data matrix. We evaluate the performance of each method for different sample sizes of the training data: $n = 128, 256, 512$, and $1024$. For every $n$, besides the training set, we generate a validation set of size $n$ and a test set including $20$K samples, and repeat the experiment 50 times. To select the hyperparameters of the models, we employ a validation set approach. We summarize the classification performances on the test set in Figure~\ref{fig:AUC_homo}. 

\begin{figure}[!htb]
\centering
\includegraphics[scale=0.8]{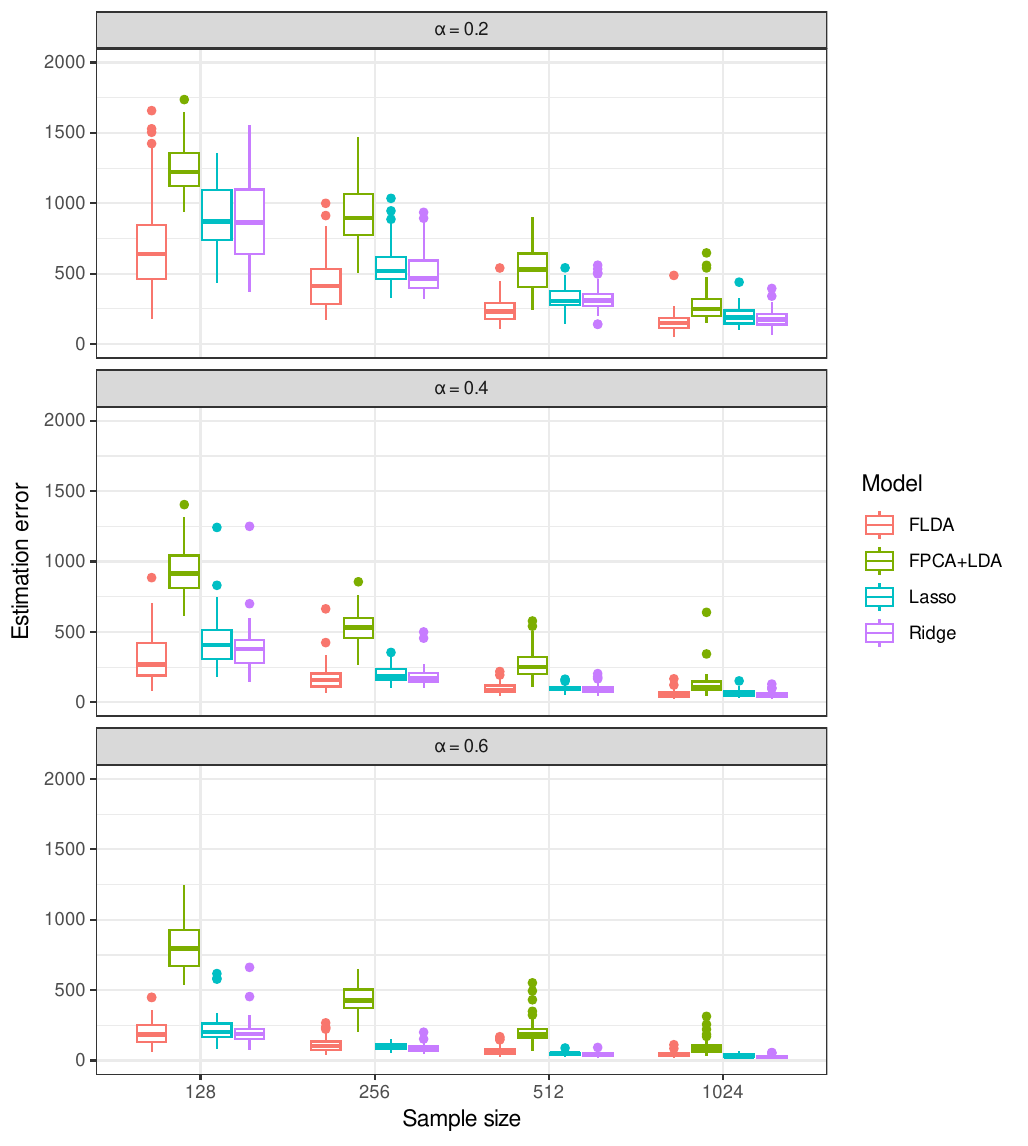}
\caption{Results of the simulation study to compare the performance of the different linear methods considered, using homogeneous covariances, for various sample sizes ($n=128,256,512,1024$) and signal-to-noise ratios ($\alpha=0.2,0.4,0.6$). Here, we measure the performance using the estimation error $\| \hat \beta -\beta^0 \|^2_{\calL^2(\M)}$, with $\hat \beta$ an appropriately normalized version of the estimate of the true functional parameter $\beta^0$.}
\label{fig:L2_homo}
\end{figure}

Figure~\ref{fig:AUC_homo} shows that by increasing $\alpha$, or the sample size $n$, the performances of all models tend to improve. This is expected, given that a larger $\alpha$ makes the classification task easier. In addition, a larger sample size allows for a more accurate estimation of the unknown model parameters. In the setting of equal within-class covariance structures, it is well known that the best classifier is linear, so it is not surprising that some of the nonlinear models, e.g., RF, SVM, and NN, show worse performances. In particular, FLDA appears to perform better than the other methods. For large sample sizes, the influence of the regularization terms becomes negligible and the difference in performance starts vanishing.

Next, we explore the performance of the different methods in the setting where the two groups have different within-group covariance structures. Specifically, we generate the data as in equation (\ref{eq:sim_generation}) with $u_{ij}$ zero-mean independent random variables distributed according to a normal with variance $\sigma_j^2$, decreasing in $j$, and $w_{ij}$ zero-mean independent random variables distributed according to a normal with variance $\sigma_{40-j}^2$, and therefore, increasing in $j$. The results of the simulations are shown in Figure~\ref{fig:AUC_het}.

In this setting, it is well known that the best classifier is not linear. Hence, it is not surprising that the linear methods tend to perform worse than the nonlinear ones and that this difference does not vanish by simply increasing the sample size. Moreover, a larger $\alpha$ leads to better performance across all the tested methods. Specifically, the proposed FQDA model outperforms all the linear models, but other non-linear approaches, such as SVM, perform even better. Importantly, the proposed FLDA model performs best among the linear models even when it is misspecified.

Classification performance is not the only relevant metric. In the application motivating this work, we are particularly interested in accurately estimating the classification rule, which for linear models is fully described by a parameter $\hat \beta$. Therefore, in the setting of homogeneous covariances, we compare the performance of the different models with respect to the metric $\| \hat \beta -\beta^0 \|^2_{\calL^2(\M)}$, with $\hat \beta$ an appropriately normalized version of the estimate of the true functional parameter $\beta^0$. Note that the estimates from a linear discriminant analysis and a logistic regression model can be compared, as they assume the same model but estimate their parameters differently \parencite{hastie2009elements}. For the standard multivariate methods, $\hat \beta$ is constructed by interpolating a piecewise linear function to its estimated discrete counterpart. The results, shown in Figure~\ref{fig:L2_homo}, indicate that the proposed FLDA model does not only yield more accurate predictions but also more accurate estimates of the underlying functional parameters.

\bigskip


\section{Proofs}\label{sec:proofs}
Recall that we denote by $C$ the covariance function of the functional predictor and by $\hat C$ its empirical counterpart. Moreover, $K_\M$ denotes the evaluation kernel of the Sobolev space $\calW^2(\M)$ endowed with the norm $\left( \|\Delta_{\M} \beta\|^2_{\calL^2(\M)} + \epsilon \|\beta\|^2_{\calL^2(\M)} \right)^{1/2}$. The sandwich operator $T$ is defined as $T = L_{K_\M}^{\frac{1}{2}} L_{C} L_{K_\M}^{\frac{1}{2}}$. To simplify the notation, we drop the subscripts from $\|\cdot \|_{\calL^2(\M)}$ and $\langle \cdot, \cdot \rangle_{\calL^2(\M)}$. We define $\| A \|_{\op} = \sup_{f : \|f\| =1} \|Af\|$ to be the operator norm of a linear operator $A: \calL^2(\M) \ra \calL^2(\M)$.

Thanks to the fact that $L_{K_{\M}}^{\frac{1}{2}}(\calL^2(\M)) = \calW^2(\M)$ \parencite{cucker2002mathematical}, it is clear that we can reformulate the problem in equation (\ref{eq:model_ls_univ}) as

\begin{equation}
\text{minimize}_{f}  \frac{1}{n} \sum_{i=1}^n \left( y_i -  \langle x_i, L_{K_\M}^{\frac{1}{2}} f \rangle \right)^2 +  \lambda \|f \|^2, 
\end{equation}
whose solution $\hat f_{\lambda}$ is given by
\begin{equation}
\hat f_{\lambda} = (T_n + \lambda I)^{-1} \frac{1}{n} \sum_{i=1}^n y_i L_{K_\M}^{\frac{1}{2}} x_i,
\end{equation}
where $T_n = L_{K_\M}^{\frac{1}{2}} L_{\hat C} L_{K_\M}^{\frac{1}{2}}$. Moreover, the out-of-sample risk $\E^* \left[ \langle X^*, \beta^0 - \hat \beta \rangle \right]^2$ can be easily rewritten as $\left\| T^{\frac{1}{2}}(\hat f_{\lambda} - f_0) \right\|^2$, with $f_0 \in \calL^2(\M)$ such that $f_0 = T^{-1} L_{K_\M}^{\frac{1}{2}} (\mu_2 - \mu_1)$.

Observe that
\begin{align*}
& T^{\frac{1}{2}}(\hat f_{\lambda} - f_0) \\
& = T^{\frac{1}{2}} \left[(T_n + \lambda I)^{-1} \left( \frac{1}{n} \sum_{i=1}^n y_i L_{K_\M}^{\frac{1}{2}} x_i \right) -  f_0 \right]\\
&= T^{\frac{1}{2}}\left[(T_n + \lambda I)^{-1} L_{K_\M}^{\frac{1}{2}} \left( \frac{1}{n} \sum_{i=1}^n y_i x_i -(\mu_2-\mu_1) \right)+ (T_n + \lambda I)^{-1} L_{K_\M}^{\frac{1}{2}} (\mu_2-\mu_1) - T^{-1} L_{K_\M}^{\frac{1}{2}} (\mu_2-\mu_1) \right]
\end{align*}
Let $\hat d = \frac{1}{n}\sum_{i=1}^n y_i  x_i$ and $d = \mu_2-\mu_1$, and notice that $\E[\hat d] = d + o(1)$. Then, we have
\begin{align}\label{eq:bias_var}
& \left\| T^{\frac{1}{2}}(\hat f_{\lambda} - f_0) \right\| \nonumber \\
& \leq \left\| T^{\frac{1}{2}}(T_n + \lambda I)^{-1} L_{K_\M}^{\frac{1}{2}} \left( \hat d - d \right)\right\| + \left\|T^{\frac{1}{2}} \left[ (T_n + \lambda I)^{-1} L_{K_\M}^{\frac{1}{2}} d- T^{-1} L_{K_\M}^{\frac{1}{2}} d \right]\right\| \\
& = I_1 + I_2 \nonumber
\end{align}

We first derive a bound for the term $I_1$ in equation (\ref{eq:bias_var}), i.e., the variance term, and then proceed with bounding the term $I_2$, i.e., the bias term. To accomplish this, we will also use Theorems~\ref{thm:Tdiff}-\ref{thm:Tsum}, stated in Section~\ref{sec:add_results}.

\subsection{Variance}
Simple calculations show that
\begin{align*}
I_1  & = \left\| T^{\frac{1}{2}}(T_n + \lambda I)^{-1} L_{K_\M}^{\frac{1}{2}} \left( \hat d - d \right)\right\| \\
& = \left\| T^{\frac{1}{2}}(T_n + \lambda I)^{-\frac{1}{2}}(T_n + \lambda I)^{-\frac{1}{2}} (T + \lambda I)^{\frac{1}{2}} (T + \lambda I)^{-\frac{1}{2}} L_{K_\M}^{\frac{1}{2}} \left( \hat d - d \right)\right\| \\
& \leq \left\|T^{\frac{1}{2}}(T_n + \lambda I)^{-\frac{1}{2}} \right\|_{\text{op}} \left\|(T_n + \lambda I)^{-\frac{1}{2}}(T + \lambda I)^{\frac{1}{2}} \right\|_{\text{op}} \left\| (T +\lambda I)^{-\frac{1}{2}} L^{\frac{1}{2}}_{K_\M} \left( \hat d - d \right) \right\| \\
& \leq \left\|(T + \lambda I)^{\frac{1}{2}}(T_n + \lambda I)^{-\frac{1}{2}} \right\|_{\text{op}} \left\|(T_n + \lambda I)^{-\frac{1}{2}}(T + \lambda I)^{\frac{1}{2}} \right\|_{\text{op}} \left\| (T +\lambda I)^{-\frac{1}{2}} L^{\frac{1}{2}}_{K_\M} \left( \hat d - d \right) \right\| \\
& \leq \left\|(T + \lambda I)(T_n + \lambda I)^{-1} \right\|_{\text{op}}  \left\|(T +\lambda I)^{-\frac{1}{2}} L^{\frac{1}{2}}_{K_\M} \left( \hat d - d \right) \right\|,
\end{align*}
where in the last inequality we have used $\|A^{\gamma} B^{\gamma}\|_{\op} \leq \|AB\|^{\gamma}_{\op}$, for any $0 < \gamma < 1$ \parencite{blanchard2010optimal}. We can then bound the first term thanks to Theorem~\ref{thm:Tsum}. We therefore turn to the second term. 

First observe that
\begin{align*}
\E \left[\hat d \right] = d + o(1), \qquad \E \left[(\hat d - d) \otimes (\hat d-d)\right] = \frac{1}{n} \left( C - \pi_1 \mu_1 \otimes \mu_1 - \pi_2 \mu_2 \otimes \mu_2 \right) + o(1),
\end{align*}
and, therefore,
\begin{align*}
&\langle L_{K_\M}^{\frac{1}{2}} (\hat d - d), \eta_k \rangle^2\\
& = \langle L_{K_\M}^{\frac{1}{2}} L_{\E \left[(\hat d - d) \otimes (\hat d - d) \right]} L_{K_\M}^{\frac{1}{2}} \eta_k , \eta_k \rangle\\
& = \frac{1}{n} \langle T \eta_k, \eta_k \rangle 
- \frac{\pi_1}{n} \langle  L_{K_\M}^{\frac{1}{2}} L_{\mu_1 \otimes \mu_1}  L_{K_\M}^{\frac{1}{2}} \eta_k, \eta_k \rangle 
- \frac{\pi_2}{n} \langle  L_{K_\M}^{\frac{1}{2}} L_{\mu_2 \otimes \mu_2}  L_{K_\M}^{\frac{1}{2}} \eta_k, \eta_k \rangle \\
&= \frac{1}{n} \left(\tau_k - \pi_1 \|L_{K_\M}^{\frac{1}{2}} \mu_1\|^2 - \pi_2 \|L_{K_\M}^{\frac{1}{2}} \mu_2\|^2\right),
\end{align*}
with $\{ \tau_k \}$ and $\{ \eta_k \}$ denoting the eigenvalues and eigenfunctions of $T$, respectively. As in \textcite{gaynanova2015optimal}, we ignore the bias term $o(1)$.

We then have
\begin{align*}
&\E\left\| (T +\lambda I)^{-\frac{1}{2}} L^{\frac{1}{2}}_{K_\M} \left( \hat d - d \right)\right\|^2 \\
&= \E \left\| \sum_k (T + \lambda I)^{-\frac{1}{2}} \eta_k \langle L_{K_\M}^{\frac{1}{2}} (\hat d - d), \eta_k \rangle \right\|^2 \\
& \leq \frac{1}{n} \sum_k \left( \frac{\tau_k}{\lambda + \tau_k} - \pi_1 \frac{\|L_{K_\M}^{\frac{1}{2}} \mu_1\|^2}{\lambda + \tau_k} - \pi_2 \frac{\|L_{K_\M}^{\frac{1}{2}} \mu_2\|^2}{\lambda + \tau_k}\right)\\
& \leq \frac{1}{n} \operatorname{D}(\lambda).
\end{align*}
Moreover, by appealing to the Markov inequality, we have that with confidence at least $1 - \delta/2$
\begin{equation}\label{eq:var_part2}
\left\|(T +\lambda I)^{-\frac{1}{2}} L^{\frac{1}{2}}_{K_\M} \left( \hat d - d \right)\right\|
 \leq \frac{2}{\delta} \sqrt{\frac{\operatorname{D}(\lambda)}{n}} \leq \frac{1}{\kappa \delta} B_{n,\lambda},
\end{equation}
where $B_{n,\lambda} = \frac{2 \kappa}{\sqrt{n}} \left(\frac{\kappa}{\sqrt{n \lambda}} + \sqrt{D(\lambda)} \right)$ and $\kappa$ is defined such that $\kappa^2 = \esssup \|L^{1/2}_{K_\M} X\|^2$. 

Thanks to the inequality in equation (\ref{eq:var_part2}) and Theorem~\ref{thm:Tsum}, with probability at least $1-\delta$, we have
\begin{align}\label{eq:bound_variance}
& I_1  = \left\| T^{\frac{1}{2}}(T_n + \lambda I)^{-1} L_{K_\M}^{\frac{1}{2}} \left( \frac{1}{n} \sum_{i=1}^n y_i  x_i - d \right)\right\| 
\leq \left( \frac{B_{n,\lambda} \log(2/\delta)}{\lambda} +1 \right)^2 \left( \frac{1}{2 \kappa \delta} B_{n,\lambda} \right).
\end{align}
We now turn to the bias term $I_2$ in equation (\ref{eq:bias_var}).

\subsection{Bias}

By simple calculations, we have
\begin{align*}
I_2 &= \left\|T^{\frac{1}{2}} \left[ (T_n + \lambda I)^{-1} L_{K_\M}^{\frac{1}{2}} d - T^{-1} L_{K_\M}^{\frac{1}{2}} d \right]\right\|\\ 
&= \left\| T^{\frac{1}{2}} (T_n + \lambda I)^{-1} [T - T_n +\lambda I] \right\|_{\text{op}} \|f_0\|\\
&\leq  \left\| (T+\lambda I)^{\frac{1}{2}}(T_n+\lambda I)^{-\frac{1}{2}}\right\|_{\text{op}}  \left\| (T_n+\lambda I)^{-\frac{1}{2}} (T+\lambda I)^{\frac{1}{2}}\right\|_{\text{op}} \left\|(T + \lambda I)^{-\frac{1}{2}} [T - T_n +\lambda I] \right\|_{\text{op}} \|f_0\|\\
&\leq  \left\| (T+\lambda I)(T_n+\lambda I)^{-1}\right\|_{\text{op}} \left\|(T + \lambda I)^{-\frac{1}{2}} [T - T_n +\lambda I] \right\|_{\text{op}} \|f_0\|\\
&\leq  \left\| (T+\lambda I)(T_n+\lambda I)^{-1}\right\|_{\text{op}} \left\|(T + \lambda I)^{-\frac{1}{2}} [T - T_n] \right\|_{\text{op}} \|f_0\|\\
&+  \left\| (T+\lambda I)(T_n+\lambda I)^{-1}\right\|_{\text{op}} \lambda \left\|(T + \lambda I)^{-\frac{1}{2}} \right\|_{\text{op}} \|f_0\|.
\end{align*}
Moreover, by using Theorems~\ref{thm:Tdiff}-\ref{thm:Tsum}, and the inequality
\begin{equation}
\lambda \left\|(T + \lambda I)^{-\frac{1}{2}}\right\|_{\text{op}} \leq \frac{\lambda}{\sqrt{\lambda}} = \sqrt{\lambda},
\end{equation}
we have that with probability at least $1-\delta$ 
\begin{align}\label{eq:bound_bias}
&\left\|T^{\frac{1}{2}} \left[ (T_n + \lambda I)^{-1} L_{K_\M}^{\frac{1}{2}} d - T^{-1} L_{K_\M}^{\frac{1}{2}} d \right]\right\| 
\leq \left( \frac{B_{n,\lambda} \log(4/\delta)}{\sqrt{\lambda}} + 1\right)^2 (B_{n,\lambda} \log(4/\delta) + \sqrt{\lambda}) \|f_0\|.
\end{align}

\subsection{Final rates}
Using the variance and bias bounds in equations (\ref{eq:bound_variance}) and (\ref{eq:bound_bias}), we get the following bound for the out-of-sample risk:
\begin{align}
& \left\| T^{\frac{1}{2}}(\hat f_{\lambda} - f_0) \right\|^2 \\
& \leq 2 \left\| T^{\frac{1}{2}}(T_n + \lambda I)^{-1} L_{K_\M}^{\frac{1}{2}} \left( \hat d - d \right)\right\|^2 + 2\left\|T^{\frac{1}{2}} \left[ (T_n + \lambda I)^{-1} L_{K_\M}^{\frac{1}{2}} d- T^{-1} L_{K_\M}^{\frac{1}{2}} d \right]\right\|^2 \\
&\leq 2 \left( \frac{B_{n,\lambda} \log(2/\delta)}{\sqrt{\lambda}} +1 \right)^4 \left( \frac{1}{2 \kappa \delta} B_{n,\lambda} \right)^2\\
&+ 2 \left( \frac{B_{n,\lambda} \log(4/\delta)}{\sqrt{\lambda}} + 1\right)^4 \left(B_{n,\lambda} \log(4/\delta) + \sqrt{\lambda} \right)^2 \|f_0\|^2\\
&\leq 2 \frac{\lambda}{\delta^2} \left( \frac{B_{n,\lambda} \log(2/\delta)}{\sqrt{\lambda}} +1 \right)^4 \left( \frac{1}{2 \kappa } \frac{B_{n,\lambda}}{\sqrt{\lambda}} \right)^2\\
&+ 2 \lambda \left( \frac{B_{n,\lambda} \log(4/\delta)}{\sqrt{\lambda}} + 1\right)^4 \left(\frac{B_{n,\lambda}}{\sqrt{\lambda}} \log(4/\delta) + 1 \right)^2 \|f_0\|^2\\
&\leq C \frac{\left(\log(4/\delta)\right)^6}{\delta^2} n^{-\frac{1}{1+\theta}},
\end{align}
where in the last inequality we have chosen $\lambda = n^{-\frac{1}{1+\theta}}$ and used the inequality $B_{n,\lambda} \leq 2 \kappa (\kappa+\sqrt{c}) \sqrt{\lambda}$. The constant $c$ here is from Assumption~\ref{assm:effective}.

This implies the result stated in Theorem~\ref{thm:risk_manifold}.

\subsection{Auxiliary results}\label{sec:add_results}
\begin{theorem}\label{thm:hoeff}
Let $\calH$ be a Hilbert space endowed with a norm $\|\cdot \|_{\calH}$ and let $X$ be a random variable taking values in $\calH$. Let $\xi_1, \ldots, \xi_n$ be a sequence of $n$ independent copies of $X$. Assume that $\|\xi \|_{\calH} \leq M$ (a.s.), then for $0 < \delta < 1$
\[
\left\| \frac{1}{n} \sum_{i=1}^{n} \left(\xi_i - \E[\xi] \right) \right\|^2_{\calH} \leq \frac{2M \log(2/\delta)}{n} + \sqrt{\frac{2 \E\left[\|\xi\|_\calH^2\right]\log(2/\delta)}{n}}
\]
with probability at least $1 - \delta$.
\end{theorem}
\begin{proof}
A proof can be found in \textcite{pinelis2007optimum}.
\end{proof}

Using Theorem~\ref{thm:hoeff}, it can be shown that the following two theorems hold.

\begin{theorem}\label{thm:Tdiff}
Under Assumption~\ref{assm:bounded}, for any $0 < \delta < 1$, the inequality 
\begin{equation}
\|(T + \lambda I)^{-\frac{1}{2}}(T - T_n)\|_{\op} \leq B_{n,\lambda} \log(2/\delta) 
\end{equation}
holds with confidence at least $1-\delta$.
\end{theorem}
\begin{proof}
A proof can be found in \textcite{tong2018analysis}.
\end{proof}

\begin{theorem}\label{thm:Tsum}
Under Assumption~\ref{assm:bounded}, for any $0 < \delta < 1$ with confidence at least $1-\delta$, 
\begin{equation}
\|(T + \lambda I)(T_n + \lambda I)^{-1}\|_{\op} \leq \left( \frac{B_{n,\lambda} \log(2/\delta)}{\sqrt{\lambda}} +1 \right)^2.
\end{equation}
Moreover, the confidence set is the same as the one in Theorem \ref{thm:Tdiff}.
\end{theorem}

\begin{proof}
A proof can be found in \textcite{tong2018analysis}.
\end{proof}

\subsection{Multivariate model}

Thanks to the multivariate notation introduced in Section~\ref{sec:multivariate}, we can reformulate the multivariate model in equation (\ref{eq:model_ls_biv}) as
\begin{equation}
\text{minimize}_{\vect f}  \frac{1}{n} \sum_{i=1}^n \left( y_i -  \langle \vect x_i, L_{\vect K}^{\frac{1}{2}} \vect{f} \rangle_{\calH} \right)^2 +  \lambda \| \vect f \|^2_{\calH},
\end{equation}
with $\vect x_i: = (v_i,x_i)$. Therefore, the proof of Theorem~\ref{thm:risk_biv} follows the same lines as that of Theorem~\ref{thm:risk_manifold} and is therefore omitted.

\setcounter{maxnames}{10}
\printbibliography

\end{document}